\documentclass[10pt,letterpaper]{article}
\usepackage[utf8]{inputenc}
\usepackage[margin=1in]{geometry}
\usepackage{palatino}
\usepackage{macros}
\usepackage[noblocks]{authblk}

\bibliography{refs}
\allowdisplaybreaks

\DeclareMathOperator{\SAW}{SAW}

\title{Spectral Independence in High-Dimensional Expanders and Applications to the Hardcore Model}

\author{Nima Anari}
\affil{\small Stanford University, \textsf{anari@cs.stanford.edu}}

\author{Kuikui Liu}
\author{Shayan Oveis Gharan}
\affil{\small University of Washington, \textsf{liukui17@cs.washington.edu}, \textsf{shayan@cs.washington.edu}}

\date{September 2020}

\begin{document}
\maketitle
\begin{abstract}
We say a probability distribution $\mu$ is spectrally independent if an associated pairwise influence matrix has a bounded largest eigenvalue for the distribution and all of its conditional distributions. 
We prove that if $\mu$ is spectrally independent, then the corresponding high dimensional simplicial complex is a local spectral expander.
Using a line of recent works on mixing time of high dimensional walks on simplicial complexes \cite{KM17,DK17,KO18,AL20}, this implies that the corresponding Glauber dynamics  mixes rapidly and generates (approximate) samples from $\mu$.

As an application, we show that natural Glauber dynamics mixes rapidly (in polynomial time) to generate a random independent set from the hardcore model up to the uniqueness threshold. This improves the quasi-polynomial running time of Weitz's deterministic correlation decay algorithm \cite{Wei06} for estimating the hardcore partition function, also answering a long-standing open problem of mixing time of Glauber dynamics \cite{LV97,LV99,DG00,Vig01,EHSVY16}.
\end{abstract}
\section{Introduction}
Suppose we have a ground set $[n]=\{1,\dots,n\}$ of elements.
Let $\mu:2^{[n]}\to\R_+$ be a probability distribution on subsets of $[n]$. We say $\mu$ is $d$-homogeneous, if for every $S\in\supp\{\mu\}$, we have $|S|=d$. When the choice of $\mu$ and $[n]$ are clear from context, we will write $\Pr[i] = \Pr_{S \sim \mu}[i \in S]$ and $\Pr[\overline{j}] = \Pr_{S \sim \mu}[i \notin S]$. The following definitions are crucial in our paper.

\begin{definition}[(Signed) Pairwise Influence Matrix]
Fix a distribution $\mu$ on subsets of a ground set $[n]$. We define the pairwise correlation matrix $\Psi_{\mu}\in\R^{n \times n}$ by
\begin{align*}
    \Psi_{\mu}(i,j) \overset{\defin}{=} \Pr[j \mid i] - \Pr[j \mid \overline{i}]
\end{align*}
for $i \neq j$, and $\Psi_{\mu}(i,i) = 0$ for all $i=1,\dots,n$. We refer to the entry $\Psi_{\mu}(i,j)$ as the pairwise influence of $i$ on $j$.
\end{definition}
Note this differs from existing definitions of ``influence''; see \cref{subsec:correlationinfluencematrices} for further discussion. One may also view $\Psi_{\mu}$ is a matrix of pairwise correlations.

\begin{definition}[Spectral Independence]\label{def:spectralind}
	We say a probability distribution $\mu$ on subsets of $[n]$ is {\em $\eta$-spectrally independent} if $\lambda_{\max}(\Psi_{\mu})\leq \eta$. Note that since the maximum eigenvalue is always at most the maximum absolute row/column sum, we have $\mu$ is $\eta$-spectrally independent if either $$\sum_{j \neq i} |\Psi_{\mu}(i,j)| \leq \eta \quad\quad \text{or} \quad\quad\sum_{j \neq i} |\Psi_{\mu}(j,i)| \leq \eta,$$
%	or for every $0\leq i\leq n-1$, we have
%	$$ \sum_{0\leq j< n, j\neq i} |\Pr_{S\sim\mu}[j\in | i\in S] - \Pr_{S\sim\mu}[j\in S | i\notin S]| \leq \eta,$$

	We say $\mu$ is $(\eta_0,\dots,\eta_{n-2})$-spectrally independent if $\mu$ is $\eta_0$-independent, for all $0\leq i<n$, $\{\mu | i \text{ in/out}\}$ is $\eta_1$-independent, for all $i,j$, $\{\mu | i \text{ in/out}, j \text{in/out}\}$ is $\eta_2$-independent,
	and so on.
\end{definition}
Note that for any $i$, we always have $\eta_i\leq n-i-1$; the smaller $\eta_i$'s are, the more independent $\mu$ is. Ideally, we are interested in distributions where $\eta_0,\dots,\eta_{n-2}\leq O(1)$ independent of $n$. Observe that if $\mu$ is a product distribution, then it is $(0,\dots,0)$-independent.

Let us explain a more interesting example. Recall that a probability distribution $\mu$ is {\em negatively correlated} if for all $i\neq j$, we have $\Pr[i | j] \leq  \Pr[i]$.
If $\mu$ is $d$-homogeneous and all measures obtainable from $\mu$ by conditioning are negatively correlated,  then $\mu$ is $(1,1,\dots,1)$-spectrally independent.

For a bad example, consider the distribution $\mu$ which places $1/2$ probability to both $\wrapc{1,\dots,\frac{n}{2}}$ and $\wrapc{\frac{n}{2} + 1,\dots,n}$. In this case , $\lambda_{\max}(\Psi_{\mu}) = n-1$.

Given a probability distribution $\mu$ we can define a Markov chain called the {\em Glauber dynamics} to generate samples from $\mu$ as follows: Given a set $S\in\supp\{\mu\}$, we choose a uniformly random element $i$ and we transition to
$$\begin{cases}
	S\smallsetminus \{i\} & \text{with prob } \frac{\mu(S\smallsetminus \{i\})}{\mu(S\smallsetminus \{i\}) + \mu(S\cup\{i\})}\\
	S\cup \{i\} & \text{o.w.,} 
\end{cases}$$
It turns out that this chain has the right stationary distribution.

The following is our main technical theorem.
\begin{theorem}[Main]\label{thm:spectralindmixing}
	For any $(\eta_0,\dots,\eta_{n-2})$-spectrally independent distribution $\mu:2^{[n]}\to\R_+$, the natural Glauber dynamics (defined above) has spectral gap at least
	$$ \frac{1}{n} \prod_{i=0}^{n-2} \left(1-\frac{\eta_i}{n-i-1}\right) $$
\end{theorem}

We note that prior works \cite{FM92, AOR16} show that as long as the distribution $\mu$, and all its conditional distributions, satisfying certain negative correlation properties, then a very similar Markov chain mixes rapidly. In our setting, negative correlation is equivalent to all entries of $\Psi_{\mu}$ being nonpositive. Thus, in a similar spirit to spectral negative dependence \cite{ALOV19ii}, one may view spectral independence and \cref{thm:spectralindmixing} as also relaxing these negative correlation requirements to allow for some positive correlation between elements, while still providing mixing time guarantees.

In the following sections we will explain an application of the above theorem in bounding the mixing time of the Glauber dynamics for sampling independent sets from the hardcore distribution. The proof of \cref{thm:spectralindmixing} uses recent connections developed by the authors and collaborators between analysis of Markov chains and the field of high dimensional expanders \cite{ALOV19ii}.

\subsection{Connections to High Dimensional Simplicial Complexes}
Let us first phrase our main contribution in the language of high-dimensional expanders. 
For a ground set $U = [n]$ of elements, a {\em simplicial complex} $X$ is a downward closed family of subsets of $U$. Sets in $X$ are also called {\em faces} of $X$. The dimension of a face in $X$ is its size.
For an integer $k$, we write $X(k)$ to denote all faces of $X$ of size $k$.
We say $X$ is {\em pure} if all maximal faces have the same size.
The dimension of $X$ is the size of the maximum face in $X$. For a pure $d$-dimensional simplicial complex $X$, we say $X$ is {\em $d$-partite} if $U$ can be partitioned into sets $U_1,\dots,U_d$ such that every maximal face $\sigma$ has exactly one element of each $U_i$.

We will often weight the maximal faces of a pure $d$-dimensional simplicial complex $X$ by some function $w:X(d) \rightarrow \R_{>0}$. This induces weights on all faces of $X$ via
\begin{align}\label{eq:faceweights}
    w(\tau) = \sum_{\sigma \in X(d) : \sigma \supset \tau} w(\sigma)
\end{align}
For a face $\tau$ of $X$, the {\em link} of $\tau$ is the simplicial complex $X_\tau=\{\sigma \smallsetminus \tau : \sigma \in X, \sigma \supset \tau\}$. We endow the maximal faces of $X_{\tau}$ with the weight $w_{\tau}(\sigma) = w(\tau \cup \sigma)$.

The {\em $1$-skeleton} of link $X_{\tau}$ of $\tau$ is a weighted graph defined as follows: For every element $i\in U$, such that $\{i\}\in X_\tau$ we have a vertex. We connect two vertices $i,j$ if $\{i,j\}\in X_\tau$ and the weight of the edges is $w_\tau(\{i,j\})$. We will let $P_{\tau}$ denote the simple random walk on the $1$-skeleton of $X_{\tau}$.

We also define a random walk on the maximal faces of $X$ by a two-step process. If the walk is currently at some $\sigma \in X(d)$, we transition by
\begin{enumerate}
    \item removing a uniformly random element $i \in \sigma$
    \item adding a random $j \notin \sigma \smallsetminus \{i\}$ to $\sigma \smallsetminus \{i\}$ with probability proportional to $w(\sigma \cup \{j\} \smallsetminus \{i\})$
\end{enumerate}
Note that there is always a nonzero probability staying at $\sigma$ in a given step.

The transition probability matrix $P_{d}^{\vee}$ of this random walk may be written down as
\begin{align*}
    P_{d}^{\vee}(\sigma,\sigma') &= \begin{cases}
        \sum_{\tau \subset \sigma : |\tau| = d-1} \frac{w(\sigma)}{d \cdot w(\tau)}, &\quad \text{if } \sigma = \sigma' \\
        \frac{w(\sigma')}{d \cdot w(\sigma \cap \sigma')}, &\quad \text{if } |\sigma \cap \sigma'| = d-1 \\
        0, &\quad \text{otherwise}
    \end{cases}
\end{align*}
where we recall that $w(\tau) = \sum_{\sigma \in X(d) : \sigma \supset \tau} w(\sigma)$.

Given a distribution $\mu$ on subsets of $[n]$, define a pure $n$-dimensional $n$-partite simplicial complex $X^{\mu}$ as follows: Let the ground set of elements be $\{1,\overline{1},2,\overline{2},\dots,n,\overline{n}\}$ with $n$ parts $U_{1} = \{1,\overline{1}\}, U_{2} = \{2,\overline{2}\}, \dots, U_{n} = \{n,\overline{n}\}$. For every set $S\in\supp\{\mu\}$ we add a maximal face $\sigma_{S}$ which has $i$ for every $i\in S$ and $\overline{i}$ for every $i\notin S$. We assign a weight to $\sigma_{S}$ given by $w(\sigma_S)=\mu(S)$.
We turn this into a simplicial complex by taking downward closure of all maximal faces. Note that in this case, $P_{n}^{\vee}$ describes exactly the Glauber dynamics for sampling $S \subset [n]$ with probability proportional to $\mu(S)$.

We are now ready to define the notion of high-dimensional expansion that we will use, which was first introduced in \cite{DK17, KM17, KO18, Opp18}.
\begin{definition}[Local Spectral Expander; \cite{KO18}]
Let $X$ be a pure $d$-dimensional simplex complex.
We say a face $\tau$ of $X$ is an $\alpha$-spectral expander if the second largest eigenvalue of the simple (non-lazy) random walk on the $1$-skeleton of $X_\tau$ is at most $\alpha$.
We say $X$ is $(\alpha_0,\dots,\alpha_{d-2})$-local spectral expander if for all $0\leq k\leq d-2$, every $\tau\in X(k)$ is an $\alpha_k$-spectral expander.
\end{definition}

We prove the following theorem making connection between spectral independence of probability distributions and local spectral expanders.
\begin{theorem}\label{thm:spectralindimpliesexpansion}
For any $(\eta_0,\dots,\eta_{n-2})$-spectrally independent distribution $\mu:2^{[n]}\to\R_+$, the pure $n$-dimensional $n$-partite simplicial complex $X^\mu$ is a $(\frac{\eta_0}{n-1}, \frac{\eta_1}{n-2},\dots,\frac{\eta_{n-2}}{1})$-local spectral expander.
\end{theorem}
We note that there are strong theorems in the literature of high-dimensional expanders \cite{Opp18} which show that if the $(d-2)$-dimensional faces of a pure $d$-dimensional complex $X$ are $\alpha$-spectral expanders for $\alpha\leq 1/2d$, then every face of $X$ is a $2\alpha$-spectral expander. However, such theorems fail dramatically when the $(d-2)$-dimensional faces have spectral expansion, say, $1/2$.

Here, the main new ingredient is to show that as long as the underlying distribution $\mu$ is spectrally independent for $\eta_0,\dots,\eta_{n-2}\leq O(1)$, then we get better and better spectral expansion as we go to lower dimensional faces of the underlying weighted simplicial complex $X^{\mu}$.

The key usefulness of local spectral expansion lies in the following local-to-global theorem, which may be used to bound $\lambda_{2}(P_{d}^{\vee})$. A weaker version of this result was already proved in \cite{KO18}.
\begin{theorem}[\cite{AL20}]\label{thm:localtoglobal}
Consider a pure $d$-dimensional simplicial complex $X$ with weights $w$. If $(X,w)$ is a $(\alpha_{0},\dots,\alpha_{d-2})$-local spectral expander, then
\begin{align*}
    \lambda_{2}(P_{d}^{\vee}) \leq 1 - \frac{1}{d} \prod_{k=0}^{d-2} (1 - \alpha_{k})
\end{align*}
\end{theorem}
\begin{remark}
For instance, if there is a constant $\alpha$ such that $(X,w)$ is a $\wrapp{\frac{\alpha}{d-1},\frac{\alpha}{d-2},\dots,\frac{\alpha}{2}, \frac{\alpha}{1}}$-local spectral expander, then we would obtain $\lambda_{2}(P_{d}^{\vee}) \leq 1 - \frac{1}{d^{1 + \alpha}}$. This is precisely what we do for the hardcore model.
\end{remark}
\begin{proof}[Proof of \cref{thm:spectralindmixing}]
By \cref{thm:spectralindimpliesexpansion}, spectral independence of $\mu$ implies strong local spectral expansion of $X^{\mu}$. \cref{thm:localtoglobal} then furnishes the spectral gap of the Glauber dynamics, which we recall is described by $P_{n}^{\vee}$.
\end{proof}
It now remains to prove \cref{thm:spectralindimpliesexpansion}.
\subsection{Application to Sampling from Hardcore Distribution}
Our main application of the above machinery is to generate random samples from the hardcore distribution. Given a graph $G=(V,E)$, and a parameter $\lambda>0$, sample an independent set $I$ with probability $\lambda^{|I|}/Z_{G}(\lambda)$, where  
\begin{align*}
    Z_{G}(\lambda) = \sum_{I \subset V \text{ independent}} \lambda^{|I|}
\end{align*}
 is the normalizing constant, a.k.a., the partition function. Exact computation of $Z_{G}(\lambda)$ is \#P-Hard \cite{Val79, Vad95, Gre00} even when the input graphs have special structure \cite{Vad02} and hence, we can only hope for efficient approximation algorithms.

Studying the hardcore model has been pivotal in helping us understand the relationship between phase transitions in statistical physics and phase transitions in efficient approximability. Specifically, has been known since \cite{Kel85} that there is a critical threshold $\lambda_{c}(\Delta) \overset{\defin}{=} \frac{(\Delta-1)^{\Delta-1}}{(\Delta-2)^{\Delta}} \approx \frac{e}{\Delta-2}$ for which the Gibbs distribution is unique on the infinite $\Delta$-regular tree if and only if $\lambda < \lambda_{c}(\Delta)$. The case $\lambda < \lambda_{c}(\Delta)$ exactly corresponds to the regime where the ``influence'' of a vertex $u$ on another vertex $v$ decays exponentially fast in the distance between $u,v$. This is known to physicists as the uniqueness regime for the hardcore model. On the flip side, $\lambda > \lambda_{c}(\Delta)$ exactly corresponds to the regime where long-range correlations persist in the model.

In the seminal work of Weitz \cite{Wei06}, it was shown that for any $\lambda < \lambda_{c}(\Delta)$ and fixed constant $\Delta$, there exists a deterministic fully polynomial time approximation scheme (FPTAS) for estimating $Z_{G}(\lambda)$. Immediately following, a sequence of results \cite{SS14,GGSVY14,GSV15,GSV16} beginning with the seminal work of Sly \cite{Sly10} proved a matching lower bound for the case $\lambda > \lambda_{c}(\Delta)$. There is no fully polynomial randomized approximation scheme (FPRAS) for estimating $Z_{G}(\lambda)$ on graphs of maximum degree $\leq \Delta$ when $\lambda > \lambda_{c}(\Delta)$ unless $\mathsf{NP} = \mathsf{RP}$. This rigorously established the first example where the statistical physics phase transition coincides with a computational complexity phase transition.

Weitz's algorithm is based on the correlation decay framework which was later on developed for estimating partition functions of two state spin systems \cite{LLY12, LLY13, SST14}. More recently, a new framework was established based on Barvinok's polynomial interpolation method \cite{Bar16, Bar17, PR17, PR19} where Weitz's result was re-proved using a different deterministic algorithm which only uses the knowledge of connected subgraphs of $G$ of diameter $O_{\epsilon,\delta}(\log n)$ \cite{PR17}. All of these methods suffer from a quasi-polynomial running time when the input graph has unbounded max-degree. 
Specifically, if $\lambda = (1 - \delta) \lambda_{c}(\Delta)$, then there is a constant $C(\delta)$ such that Weitz's correlation decay algorithm returns a $(1\pm\epsilon)$-multiplicative approximation of $Z_{G}(\lambda)$ in time $O\wrapp{(n/\epsilon)^{C(\delta)\log \Delta}}$. In particular, due to the exponential dependence in $\log \Delta$, Weitz's algorithm does not run in polynomial time for graphs with unbounded maximum degree.
Roughly speaking, the main difficulty is that in order to estimate the partition function, one needs to estimate the marginal probabilities of vertices within $O(1/n)$-error, and to do that one needs to look at $O(\log n)$-depth neighborhood of vertices which leads to a quasi-polynomial number of operations on graphs of max-degree polynomial in $n$.

On the other hand, it is conjectured that the natural Glauber dynamics mixes in polynomial time up to the uniqueness threshold. But to this date after a long line of works \cite{LV97,LV99,DG00,Vig01} this was only shown up to $\frac{2}{\Delta-2}$ for general graphs and up to the uniqueness threshold for special families of graphs \cite{Wei04, Wei06, RSTVY13, EHSVY16}.
We use the result of the previous sections to prove that for {\em any graph} the Glauber dynamics mix in polynomial time up to the uniqueness threshold.

%However, one key drawback of Weitz's result is the running time of the correlation decay algorithm. Specifically, if $\lambda = (1 - \delta) \lambda_{c}(\Delta)$, then there is a constant $C(\delta)$ such that Weitz's correlation decay algorithm returns a $(1\pm\epsilon)$-multiplicative approximation of $Z_{G}(\lambda)$ in time $O\wrapp{(n/\epsilon)^{C(\delta)\log \Delta}}$. In particular, due to the exponential dependence in $\log \Delta$, Weitz's algorithm does not run in polynomial time for graphs with unbounded maximum degree.

%On the other hand, Markov chain algorithms have the potential to avoid such poor dependence on $\Delta$. 
For sampling from the hardcore model, the Glauber dynamics can be described via the following two-step process. To make a transition from an independent set $I$ to another,
\begin{enumerate}
    \item Select a uniformly random vertex $v \in V$.
    \item If $v \in I$, remove $v$ from $I$ with probability $\frac{1}{1 + \lambda}$, and keep it otherwise.
    \item If $v \notin I$ and $v$ is not a neighbor of some $u \in I$, add $v$ to $I$ with probability $\frac{\lambda}{1 + \lambda}$, and leave it otherwise.
\end{enumerate}
It is clear that this process is reversible. It is also clear that this Markov chain is connected, since there is a path from every independent set to the empty independent set $\emptyset$. Hence, these dynamics have a unique stationary distribution $\pi$, and the distribution of the chain converges to stationarity in total variation distance as the number of steps goes to infinity. Finally, by checking the detailed balance condition that the stationary distribution $\pi$ of the Glauber dynamics is exactly the Gibbs distribution $\mu$. Our goal is to bound the $\epsilon$-total variation mixing time of the Glauber dynamics starting from any state $\tau$, which is given by
\begin{align*}
    t_{\tau}(\epsilon) = \min\wrapc{t \in \Z_{\geq0} : \norm{P^{t}(\tau,\cdot) - \pi}_{1}} \leq \epsilon
\end{align*}
where $P$ denotes the transition probability matrix describing the chain. Here, $P^{t}(\tau,\cdot)$ gives the distribution at time $t$ of the chain started at $\tau$.
\begin{theorem}\label{thm:hardcorespectrallyindependent}
There is a function $C:[0,1] \rightarrow \R_{>0}$ such that for every graph $G=(V,E)$ with maximum degree $\leq \Delta$, every $0 < \delta < 1$, and $\lambda = (1 - \delta) \lambda_{c}(\Delta)$, the associated hardcore distribution $\mu$ is $(\eta_{0},\dots,\eta_{n-2})$-spectrally independent where $\eta_{i} \leq \min\left\{C(\delta),\frac{\lambda}{1+\lambda}(n-i-1)\right\}$ for every $0 \leq i \leq n-2$.
\end{theorem}
Combined with \cref{thm:spectralindmixing}, we obtain fast mixing for the Glauber dynamics for sampling independent sets according to the hardcore distribution whenever $\lambda < \lambda_{c}(\Delta)$ (for the precise mixing time, see \cref{rem:mixingtime}). 
%In particular, we obtain an efficient algorithm even without the assumption $\Delta \leq O(1)$.
\begin{corollary}\label{cor:fprasunboundeddegree}
For every $\delta > 0$, there exists a fully polynomial randomized approximation scheme for estimating $Z_{G}(\lambda)$ at $\lambda = (1 - \delta) \lambda_{c}(\Delta)$ on any graph $G$ with maximum degree $\leq \Delta$.
\end{corollary}
\begin{remark}\label{rem:mixingtime}
For $0 < \delta < 1$, $\lambda = (1 - \delta) \lambda_{c}(\Delta)$ and a graph $G=(V,E)$  maximum degree $\Delta$, the Glauber dynamics from any starting state $\tau$ has mixing time
\begin{align*}
    t_{\tau}(\epsilon) \leq O\wrapp{((1 + \lambda) \cdot n)^{1 + C(\delta)} \cdot \log\wrapp{\frac{1}{\epsilon \cdot \mu(\tau)}}}
\end{align*}
 To be explicit, the constant $C(\delta)$ obeys the bound
\begin{align*}
    C(\delta) \leq \exp(O(1/\delta))
\end{align*}

Note $\mu(\emptyset) = \frac{1}{Z_{G}(\lambda)} \geq \frac{1}{(1 + \lambda)^{n}}$ so that
\begin{align*}
    t_{\emptyset}(\epsilon) \leq O\wrapp{((1 + \lambda) \cdot n)^{1 + C(\delta)} \cdot \log\wrapp{\frac{1}{\epsilon \cdot \mu(\emptyset)}}} \leq O\wrapp{n^{2 + C(\delta)} \cdot \log\wrapp{\frac{1}{\epsilon}}}
\end{align*}
The key advantage of our result is that the running time has no dependence on $\Delta$. Furthermore, $\lambda_{c}(\Delta) \leq 4$ for all $\Delta \geq 3$ so we may treat $\lambda$ as bounded above by a constant. Hence, only the gap parameter $\delta$ matters.
\end{remark}

\subsection{Related Works}
The question of building deterministic approximation algorithms for estimating $Z_{G}(\lambda)$ on bounded degree graphs has been settled. \cite{Wei06} proved that there is an FPTAS on graphs of maximum degree $\leq\Delta$ whenever $\lambda < \lambda_{c}(\Delta)$. \cite{HSV18} extends this result to estimating the multivariate independence polynomial, and \cite{PR19} proves the existence of a zero-free region around $[0,\lambda_{c}(\Delta))$, which makes Barvinok's polynomial interpolation technique \cite{Bar17} applicable to estimating $Z_{G}(\lambda)$; see \cite{PR17}. We note that the Bethe approximation for estimating $Z_{G}(\lambda)$ has also been studied in \cite{CCGSS11}.

For studying the mixing time of the Glauber dynamics in the uniqueness regime, there has been a long line of work starting with \cite{LV97, LV99, DG00, Vig01}. For general graphs, the state-of-the-art was given by \cite{Vig01}, which showed the Glauber dynamics mixes in $O(n\log n)$ steps when $\lambda < \frac{2}{\Delta-2}$. A more recent result of \cite{EHSVY16} shows that for any $0 < \delta < 1$, there is a $\Delta_{0}(\delta)$ such that for any $\Delta \geq \Delta_{0}(\delta)$ and $\lambda = (1 - \delta)\lambda_{c}(\Delta)$, the Glauber dynamics mixes in $O(n\log n)$ steps for graphs with maximum degree $\Delta$ and girth $\geq 7$.

More is known for restricted families of graphs. The hardcore distribution over independent sets of the line graph $L(G)$ of a graph $G$ is equivalent to the monomer-dimer distribution over matchings of $G$ itself. Here, the Glauber dynamics is known to mix in time $O(\lambda^{3}mn^{2}\log n)$ time \cite{JS89b}. \cite{BGKNT07, SSSY15} give deterministic FPTAS for this problem in the full range of $\lambda$ on bounded-degree graphs. It is proved in \cite{Wei06} that the Glauber dynamics mixes in $O(n^{2})$ steps for any $\lambda < \lambda_{c}(\Delta)$ when the input graph has maximum degree $\leq \Delta$ and satisfies subexponential growth. This encompasses, for instance, the integer lattices $\Z^{d}$. On such lattices, there is an equivalence between strong spatial mixing and optimal mixing of the Glauber dynamics \cite{DSVW02, Wei04}. \cite{MSW03, MSW04, Wei04} obtained rapid mixing for trees, and \cite{Hay06} obtained rapid mixing for planar graphs. For graphs of large girth, \cite{HV05} studies the mixing time of the Glauber dynamics and \cite{BG08} studies deterministic correlation decay algorithms. In the case of the square grid $\Z^{2}$, we have a more precise understanding \cite{VVY13, RSTVY13, BGRT13, BCGRT19}. \cite{MS08, MS13, SSY13, SSSY15} study the case of $G(n,d/n)$ random graphs, or more generally graphs with bounded connective constant.

On the hardness side, exact computation of $Z_{G}(\lambda)$ is known to be \#P-Hard \cite{Val79, Vad95, Gre00}, even for very restricted families of graphs \cite{Vad02}. For hardness of approximation, \cite{LV97} showed there exists a constant $c > 0$ such that there is no FPRAS for estimating $Z_{G}(1)$ when $\lambda > c/\Delta$ unless $\mathsf{NP} = \mathsf{RP}$. For the case of evaluating $Z_{G}(1)$, this was improved in \cite{DFJ02}, which showed that there is no FPRAS for estimating $Z_{G}(1)$ on graphs with maximum degree exceeding $25$ unless $\mathsf{NP} = \mathsf{RP}$. \cite{DFJ02} further showed that the Glauber dynamics has exponential mixing time for $\Delta \geq 6$. \cite{MWW07} provided further evidence the Markov chain techniques are likely to fail for sampling from the Gibbs distribution when $\lambda > \lambda_{c}(\Delta)$. These results were dramatically improved in the work of \cite{Sly10} (and further refined by follow-up works \cite{SS14,GGSVY14,GSV15,GSV16}), which showed that unless $\mathsf{NP} = \mathsf{RP}$, there is no FPRAS for estimating $Z_{G}(\lambda)$ on graphs of maximum degree $\leq\Delta$ when $\lambda > \lambda_{c}(\Delta)$.

There are also many works extending results for the hardcore model to general antiferromagnetic two-state spin systems. For antiferromagnetic Ising models in the uniqueness regime, there are FPTAS based on both correlation decay \cite{SST14} (see also \cite{BLZ11} for a special case) and polynomial interpolation \cite{LSS19, Liu19, SS19}. Combined with algorithms for the hardcore model, these give FPTAS for all antiferromagnetic two-state spin systems via reductions described for instance in \cite{SST14}. In a more direct fashion, \cite{LLY12, LLY13} give deterministic correlation decay algorithms for all antiferromagnetic two-state spin systems up to the uniqueness threshold. \cite{GJP03} analyze the corresponding Glauber dynamics via the path coupling method, but do not obtain rapid mixing in entire uniqueness regime. \cite{GJP03} provide hardness of approximation for a certain range of edge activities $\beta,\gamma$. \cite{SS14, GGSVY14, GSV15, GSV16} extend these hardness results to all antiferromagnetic two-state spin systems in the nonuniqueness regime.

\subsection{Relation to Existing Definitions of Influence}\label{subsec:correlationinfluencematrices}
Fix a distribution $\mu$ on subsets of a ground set $[n]$. Recall that our pairwise influence matrix of $\mu$ is described by
\begin{align*}
    \Psi_{\mu}(i,j) =  \Pr[j \mid i] - \Pr[j \mid \overline{i}]
\end{align*}
for all $i,j$, with zeros on the diagonal. For spin systems, this definition is reminiscent but different from the influence matrix method used in \cite{Hay06, DGJ09} and the works \cite{Dob70, DS85i, DS85ii, DS87} on variants of the Dobrushin condition. Specifically, the $(i,j)$th entry of the Dobrushin influence matrix (also known as the matrix of ``dependencies''/``interdependencies'') considered in these prior works is given by the maximum absolute difference $\abs{\Pr[j \mid i, \tau] - \Pr[j \mid \overline{i},\tau]}$ over all partial assignments $\tau$ of the remaining ground elements excluding $i,j$. In the case of the hardcore distribution for an input graph $G=(V,E)$ with fugacity $\lambda > 0$, this influence matrix is exactly $\frac{\lambda}{1+\lambda}A$, where $A$ is the adjacency matrix of $G$ (see, for instance, \cite{Hay06}).

On the other hand, our pairwise influence matrix $\Psi_{\mu}$ may have nonzero entries for $u,v \in V$ not connected by an edge, since $\Psi_{\mu}(u,v)$ considers the marginal of $v$ conditioned only on $u$ or $\overline{u}$, with the assignment for other elements left undetermined. Furthermore, our method requires understanding exponentially many pairwise influence matrices, one for each conditional distribution, while all variants of the Dobrushin condition only require analyzing a single Dobrushin influence matrix.
%In fact, while one may think of $\Psi_{\mu}$ as also a kind of ``influence matrix'', it is actually more similar to a covariance/correlation matrix (hence, the name). This can be made precise in the sense that there exists a diagonal matrix $D \in \R^{n \times n}$ such that $D\Psi$ is the covariance matrix of $S \sim \mu$ viewed as an $n$-dimensional $\{0,1\}$ indicator vector. The entries of $D$ are given by $D(i,i) = \Pr[i] \cdot \Pr[\overline{i}]$.

\subsection{Subsequent Works}
Finally, we mention several follow-up works applying and extending the notion of spectral independence we introduce this paper. The first is the work by \cite{CLV20}, where they obtained rapid mixing of the Glauber dynamics for all two-state spin systems in the correlation decay regime. \cite{CGSV20, FGYZ20} extended our notion of spectral independence to multi-state spin systems, including $q$-colorings of a graph. The extension of spectral independence in \cite{FGYZ20} is more reminiscent of the Dobrushin uniqueness condition \cite{Dob70, DS85i, DS85ii, DS87} discussed in \cref{subsec:correlationinfluencematrices}, while the extension in \cite{CGSV20} is precisely the observation in \cref{rem:partitegeneralization}. Both approaches were applied to obtain rapid mixing of the Glauber dynamics for sampling $q$-colorings on triangle-free graphs of maximum degree $\Delta$ when $q > \alpha^{*}\Delta$, where $\alpha^{*} \approx 1.763$ is the unique solution to the equation $xe^{-1/x} = 1$.

\subsection{Proof Overview}
For a face $\sigma$ of $X^\mu$, recall $P_{\sigma}$ denotes the transition probability matrix of the simple random walk on the $1$-skeleton of $X^{\mu}_\sigma$. Our first technical contribution is the following.
\begin{theorem}\label{thm:correlationspecradius}
For every distribution $\mu$ over subsets of a ground set $[n]$, the eigenvalues of $\Psi_{\mu}$ are real. Furthermore, we have the identity $\lambda_{2}(P_{\emptyset}) = \frac{1}{n-1} \cdot \lambda_{\max}(\Psi_{\mu})$.
\end{theorem}
Given this, we may now prove \cref{thm:spectralindimpliesexpansion}.
\begin{proof}[Proof of \cref{thm:spectralindimpliesexpansion}]
Since \cref{thm:correlationspecradius} holds for any distribution $\mu$, it in particular holds for all conditional distributions of $\mu$. Now, observe that conditioning on an element $i$ being ``in'' corresponds exactly to taking the link of $X^{\mu}$ w.r.t. $i$. Similarly, conditioning on an element $i$ being ``out'' corresponds exactly to taking the link of $X^{\mu}$ w.r.t. $\overline{i}$. The result then follows by definition of spectral independence and local spectral expansion.
\end{proof}

The proof of \cref{thm:correlationspecradius} hinges on the observation that for each element $i \in [n]$, no face of $X^{\mu}$ can contain both $i$ and $\overline{i}$. In particular, there is no edge connecting $i$ and $\overline{i}$ in the $1$-skeleton of $X^{\mu}$, for each $i$. Thus, there are $n$ parts, one corresponding to each element of $[n]$, such that all edges only go between parts. This \textit{$n$-partite} structure of the $1$-skeleton of $X^{\mu}$ induces $n-1$ additional ``trivial'' eigenvalues, besides the trivial eigenvalue of $1$, in the transition matrix $P_{\emptyset}$. This is, in fact, a generalization of the fact that the transition matrix of a bipartite graph always also has eigenvalue $-1$. We show that $\Psi_{\mu}$ is essentially equal to $P_{\emptyset}$ projected away from these $n$ trivial eigenvalues; see \cref{claim:PemptysetMemptyset} and \cref{claim:MemptysetPsi} in \cref{sec:correlationspecradius} for more details.

%We prove the above theorem in \cref{??}. Note that $P_{\emptyset} \in \R^{2n \times 2n}$ while $\Psi_{\mu} \in \R^{n \times n}$. Hence, to prove this result, we analyze the structure of the right eigenvectors of $P_{\emptyset}$. Specifically, we show these satisfy a certain system of $n$ linear equations due to the partiteness of $X^{\mu}$. 

We apply these results to the hardcore distribution over independent sets of an input graph $G=(V,E)$. \cref{thm:correlationspecradius} tells us that to bound $\lambda_{2}(P_{\emptyset})$, it suffices to bound $\lambda_{\max}(\Psi_{\mu})$. We show how to bound $\lambda_{\max}(\Psi_{\mu})$ by bounding $\sum_{u \in V : u \neq v} \abs{\Psi_{\mu}(u,v)}$ for any vertex $v \in V$. In particular, we have the following two bounds.
\begin{lemma}\label{lem:totalinfboundlambda}
Consider the hardcore distribution $\mu$ on independent sets of a graph $G=(V,E)$ on $n$ vertices. Then for every $v \in V$, and every $\lambda > 0$, we have the bound
\begin{align*}
    \sum_{u \in V : u \neq v} \abs{\Psi_{\mu}(u,v)} \leq \frac{\lambda}{1 + \lambda} \cdot (n-1)
\end{align*}
\end{lemma}
\begin{proof}
Observe that the maximum probability that a given vertex is placed in a random independent set is at most $\frac{\lambda}{1 + \lambda}$. In particular, $\Pr[v \mid u], \Pr[v \mid \overline{u}] \in \wrapb{0, \frac{\lambda}{1 + \lambda}}$ so that $\abs{\Psi_{\mu}(u,v)} \leq \frac{\lambda}{1 + \lambda}$ for every $u \neq v$. The claim follows.
\end{proof}
\begin{theorem}\label{thm:hardcoretotalinfbound}
There exists a function $C : [0,1] \rightarrow \R_{>0}$ such that for every graph $G=(V,E)$ of maximum degree $\leq \Delta$, every vertex $v \in V$, every $0 < \delta < 1$, and $\lambda = (1 - \delta)\lambda_{c}(\Delta)$, we have the following bound, 
%which we we emphasize is independent of $n = |V|$:
\begin{align*}
    \sum_{u \in V : u \neq v} \abs{\Psi_{\mu}(u,v)} \leq C(\delta)
\end{align*}
To be explicit, $C(\delta)$ satisfies $C(\delta) \leq \exp(O(1/\delta))$.
\end{theorem}
\begin{remark}
We believe $C(\delta) \leq O(1/\delta)$ is possible, which we show is tight (see \cref{app:infonregtree}). We leave this as an open problem. We note that in follow-up work, \cite{CLV20} shows that $\sum_{u \in V : u \neq v} \abs{\Psi_{\mu}(v,u)} \leq O(1/\delta)$. In other words, they analyze the total pairwise influence of a vertex, while we analyze the total pairwise influence on a vertex.
\end{remark}

The key here is that we only need to understand the total sum of correlations between pairs of vertices. This is in contrast to strong spatial mixing results, where one has to analyze the correlation of any subset of vertices on another given vertex. 
%Note that one can immediately apply strong spatial mixing to obtain a bound of
%\begin{align*}
 %   \sum_{u \in V : u \neq v} \abs{\Psi_{\mu}(u,v)} \leq C \cdot \sum_{u \in V : u \neq v} \alpha^{d(u,v)} \leq C \cdot \sum_{\ell=1}^{\diam(G)} \abs{B(v,\ell) \smallsetminus B(v,\ell-1)} \cdot \alpha^{\ell}
%\end{align*}
%where $B(v,\ell) = \{u \in V : d(u,v) \leq \ell\}$ denotes the ball of radius $\ell$. On graphs whose balls $B(v,\ell)$ grow sufficiently slowly (subexponentially in the radius) such as the integer lattice $\Z^{d}$, the right-hand side here is bounded by a constant as desired. However, for most graphs (such as expander graphs), balls grow exponentially fast, and so in general we cannot hope for an upper bound that is independent of $n = |V|$ using this simple approach.

To prove \cref{thm:hardcoretotalinfbound}, first, we take advantage of the self-avoiding walk tree construction introduced in \cite{Wei06} to reduce to a problem on trees.
Then, we give an method to decouple the influence of a set of vertices $S$ on a vertex $v$ into the sum of the single-vertex pairwise influences of each $u \in S$ on $v$.  The primary takeaway from these two steps is that it suffices to control the total pairwise influence of vertices on the root in any rooted tree of maximum degree $\leq\Delta$.

To control correlations between vertices and the root, we leverage the well-known tree recursion, which expresses the marginal of the root in terms of the marginals of its children. We amortize the total pairwise influence of all vertices at a fixed distance from the root using the potential method \cite{LLY12, LLY13, RSTVY13, SSY13, SST14, SSSY15}; we refer to \cite{Sri14} for further discussion of the potential method. This allows us to show a strong kind of correlation decay, where the total pairwise influence of all vertices at a fixed distance decays as the distance grows. After sharing a preliminary draft of this paper, it was pointed out to us by Eric Vigoda and Zongchen Chen that the notion of correlation decay we prove is very similar to the notion of \textit{aggregate strong spatial mixing} (for trees) studied in \cite{MS13, BCV20}.

Now, observe that \cref{thm:hardcorespectrallyindependent} simply follows from \cref{lem:totalinfboundlambda} and \cref{thm:hardcorespectrallyindependent}.
%\begin{proof}[Proof of \cref{thm:hardcorespectrallyindependent}]
%Note that \cref{thm:hardcoretotalinfbound} is vacuous if $C(\delta) > n - 1$, which holds for sufficiently small $n$ for any $\delta$. In this case, we employ \cref{lem:totalinfboundlambda}. These observations combined with \cref{thm:correlationspecradius} immediately gives \cref{thm:hardcorespectrallyindependent}.
%\end{proof}
As a consequence, %for the case of the hardcore model, 
all we are left to do is to prove \cref{thm:correlationspecradius} and \cref{thm:hardcoretotalinfbound}, which we do in the remainder of the paper. 
%In the appendices, show how to prove \cref{thm:antiferrospectrallyindep}. 
\subsection{Structure of the Paper}
In \cref{sec:prelim}, we review necessary background in the theory of Markov chains, and correlation decay. In \cref{sec:correlationspecradius}, we prove \cref{thm:correlationspecradius}. In \cref{sec:decoupling} and \cref{sec:potentialmethodratiosinfdecay}, we prove \cref{thm:hardcoretotalinfbound}. 
%and its slight generalization to two-state spin systems with $\beta = 0$ and arbitrary $\gamma > 0$. This also proves \cref{thm:antiferrospectrallyindep} in the case $\beta = 0$ and $\gamma > 0$ is arbitrary. We analyze the case $\beta > 0$ in \cref{subapp:betagammabig} and \cref{subapp:betagammasmall} to finish the proof of \cref{thm:antiferrospectrallyindep}. 
In \cref{app:infonregtree}, we show that on the infinite $\Delta$-regular tree, the hardcore distribution $\mu$ over independent sets satisfies $\lambda_{\max}(\Psi_{\mu}) = \Theta(1/\delta)$ when $\lambda < (1-\delta)\lambda_{c}(\Delta)$.

\subsection{Acknowledgements}
We thank Vedat Levi Alev and Lap Chi Lau for stimulating discussions, and for sharing with us a preliminary draft of their results \cite{AL20}. We also thank Dorna Abdolazimi, Zongchen Chen, and Chihao Zhang for helpful comments, and for pointing out bugs in earlier versions of this paper. Finally, we thank Eric Vigoda, Lap Chi Lau and the anonymous referees for delivering helpful feedback on a preliminary draft of this paper.

Shayan Oveis Gharan and Kuikui Liu are supported by the NSF grants CCF-1552097 and CCF-1907845, and ONR-YIP grant N00014-17-1-2429.

\section{Preliminaries}\label{sec:prelim}
	First, let us establish some notational conventions. Unless otherwise specified, all logarithms are in base $e$. All vectors are assumed to be column vectors. For two vectors $\phi, \psi\in \R^n$, we use $\langle \phi, \psi\rangle$ to denote the standard Euclidean inner product between $\phi$ and $\psi$. We use $\R_{>0}$ and $\R_{\geqslant 0}$ to denote the set of positive and nonnegative real numbers, respectively, and $[n]$ to denote $\{1,\dots,n\}$. If $\phi \in \R^{n}$, then $\diag(\phi) \in \R^{n \times n}$ denotes the diagonal matrix with the entries of $\phi$ on the diagonal.

\subsection{Linear Algebra}
    Throughout, if $A$ is a matrix with real eigenvalues, then we write $\lambda_{n}(A) \leq \dots \leq \lambda_{1}(A)$ for its eigenvalues.
    \begin{lemma}\label{lem:maxabsrowcolsum}
    For a square matrix $A$, which need not have real eigenvalues, we have the inequalities
    \begin{align*}
        \lambda_{\max}(A) \leq \max_{u=1,\dots,n} \sum_{j=1}^{n} |A_{ij}| \quad\quad\quad\quad \lambda_{\max}(A) \leq \max_{j=1,\dots,n} \sum_{i=1}^{n} |A_{ij}|
    \end{align*}
    \end{lemma}
	\begin{lemma}\label{lem:ABBAeigs}
	Let $A \in \R^{n \times m}, B \in \R^{m \times n}$ where $m \geq n$. Then the spectrum of $BA$ (as a multiset) is precisely the union of the spectrum of $AB$ (as a multiset) with $m - n$ copies of $0$.
	\end{lemma}
	\subsection{Markov Chains and Random Walks}\label{subsec:randomwalks}
	For this paper, we consider a Markov chain as a triple $(\Omega,P,\pi)$ where $\Omega$ denotes a (finite) state space, $P\in \R_{\geqslant 0}^{\Omega\times \Omega}$ denotes a transition probability matrix and $\pi \in \R_{\geqslant 0}^{\Omega}$ denotes a stationary distribution of the chain (which will be unique for all chains we consider). For  $\tau,\sigma \in \Omega$, we use $P(\tau,\sigma)$ to denote the corresponding entry of $P$, which is the probability of moving from $\tau$ to $\sigma$. 

	We say a Markov chain is $\epsilon$-lazy if for any state $\tau \in \Omega$, $P(\tau,\tau) \geqslant \epsilon$. A chain $(\Omega, P, \pi)$ is \emph{reversible} if there is a nonzero nonnegative function $f:\Omega\to\R_{\geqslant0}$ such that for any pair of states $\tau, \sigma \in \Omega$, 
	\[f(\tau) P(\tau,\sigma) = f(\sigma)P(\sigma,\tau).\]
	If this condition is satisfied, then $f$ is proportional to a stationary distribution of the chain. In this paper we only work with reversible Markov chains. Note that being reversible means that the transition matrix $P$ is self-adjoint w.r.t.\ the following inner product, defined for $\phi,\psi\in \R^\Omega$:
	\[ \langle\phi,\psi\rangle_{f}=\sum_{x\in\Omega} f(x)\phi(x)\psi(x).  \]

	Reversible Markov chains can be realized as random walks on weighted graphs. Given a weighted graph $G=(V,E,w)$ where every edge $e\in E$ has weight $w(e)$, the non-lazy \emph{simple} random walk on $G$ is the Markov chain that from any vertex $u\in V$ chooses an edge $e=\wrapc{u,v}$ with probability proportional to $w(e)$ and jumps to $v$. We can make this walk $\epsilon$-\emph{lazy} by staying at every vertex with probability $\epsilon$. It turns out that if $G$ is connected, then the walk has a unique stationary distribution where $\pi(u)\propto w(u)$, where $w(u)=\sum_{v\sim u} w(\wrapc{u,v})$ is the weighted degree of $u$.

	For any reversible Markov chain $(\Omega, P, \pi)$, the largest eigenvalue of $P$ is $1$. We let $\lambda^*(P)$ denote the second largest eigenvalue of $P$ in absolute value. That is, if $-1\leq \lambda_n\leq \dots\leq \lambda_1=1$ are the eigenvalues of $P$, then $\lambda^*(P)$ equals $\max\wrapc{\abs{\lambda_2},\abs{\lambda_n}}$.
	\begin{theorem}[\cite{DS91}]\label{thm:mixingtime}
		For any reversible irreducible  Markov chain $(\Omega, P, \pi)$,  $\epsilon>0$, and any starting state $\tau\in \Omega$,
	\[ t_\tau(\epsilon)  \ \leq  \ \frac1{1-\lambda^*(P)}\cdot \log\wrapp{\frac{1}{\epsilon\cdot \pi(\tau)}}.\]
	\end{theorem}

\subsection{Tree Recurrences for Hardcore Model}
Fix a tree $T$ rooted at some vertex $r$. For a vertex $v$ in $T$, let $\ell(v)$ denote its distance from the root $r$. We will sometimes refer to it as the ``level'' which contains $v$. For a level $\ell$, let $L_{r}(\ell) = \{v \in T : \ell(v) = \ell\}$. For a vertex $u \in T$, we will write $T_{u}$ for the subtree of $T$ rooted at $u$.

A key tool we will need to analyze the hardcore model on trees is given by the tree recurrence. To describe the tree recurrence, we need to consider a change of variables w.r.t. the marginal probabilities. Fix a tree $T$ arbitrarily rooted at some vertex $r \in T$, and an arbitrary boundary condition $p:A \rightarrow [0,1]$ for a subset of remaining vertices $A$. We write the ratio of conditional probabilities as
\begin{align*}
    R_{T,r}^{p} = \frac{\Pr[r \mid p]}{\Pr[\overline{r} \mid p]} = \frac{\Pr[r \mid p]}{1 - \Pr[r \mid p]}
\end{align*}
Here, we think of the function $p:A \rightarrow [0,1]$ as fixing the marginal probability of vertices $v \in A$ to $p(v)$. In the special case where $p$ maps all vertices of $A$ to $0$ or $1$, then $p$ is really a boundary condition in the traditional sense, as $p$ is pinning the vertices of $A$ to be in/out. However, later on, we will need the additional flexibility of pinning the marginal of $v \in A$ to a specific value $p(v) \in [0,1]$

With this notation in hand, we may write the tree recurrence for the hardcore model as
\begin{align}\label{eq:treerecurrence}
    R_{T,r}^{p} = F(R_{T_{u},u}^{p} : u \in L_{r}(1)) \overset{\defin}{=} \lambda \prod_{u \in L_{r}(1)} \frac{1}{R_{T_{u},u}^{p} + 1}
\end{align}
Here, we make a slight abuse of notation by writing $p$ even when considering a subtree $T_{u}$; this should be understood as the restriction of $p$ to this subtree. We drop the superscript when $p$ is empty; we also drop the subscript $T$ when the tree is clear from context. Note that this tree recurrence naturally leads to a simple polynomial-time dynamic programming algorithm for exactly computing $Z_{G}(\lambda)$ on any tree.

In the case of a depth-$\ell$ complete $d$-ary tree rooted at $r$ with no boundary conditions, all of the $R_{u}^{\sigma}$ are the same, as the only relevant parameter is the depth of tree. In this case, the tree recurrence simplifies to a univariate recurrence given by
\begin{align*}
    f_{d}(R) = \lambda \wrapp{\frac{1}{R + 1}}^{d}
\end{align*}
It turns out that there exists a unique fixed point to $f_{d}(\cdot)$ we call this $\hat{R}_d$, i.e.,
$$ \hat{R}_d = \lambda \left(\frac{1}{\hat{R}_d+1}\right)^d,$$
see \cite{LLY13} and references therein.
%Note that since $f_{d}$ is strictly decreasing, $g_{d}(R) \overset{\defin}{=} f_{d}(R) - R$ is strictly decreasing, and satisfies $g_{d}(0) = f_{d}(r) = \lambda$ while $\lim_{R \rightarrow \infty} g_{d}(R) = -\infty$. By the Intermediate Value Theorem, there exists $\hat{R}_{d}$ such that $g_{d}(\hat{R}_{d}) = 0$. Note this $\hat{R}_{d}$ is unique by strict monotonicity of $g_{d}$. Translated into the $f_{d}$, we have $\hat{R}_{d}$ is the unique fixed point of $f_{d}$, i.e. $f_{d}(\hat{R}_{d}) = \hat{R}_{d}$. It turns out the properties of this fixed point $\hat{R}_{d}$ and the behavior of $f_{d}$ at $\hat{R}_{d}$ govern the correlation decay properties of the hardcore model, as we will see in the following subsection.

\subsection{Correlation Decay and Weitz's Self-Avoiding Walk Tree}
In this subsection, we introduce the necessary notation for describing the correlation decay property for spin systems. We begin by rigorously defining correlation decay for a general distribution $\mu$ over subsets of a ground set $[n]$.
\begin{definition}[Weak Spatial Mixing]
Fix a metric $d:[n] \times [n] \rightarrow \R_{\geq0}$. We say a distribution $\mu$ on subsets of $[n]$ exhibits weak spatial mixing w.r.t. $d$ with rate $0 < \alpha < 1$ and constant $C > 0$ if for every $i \in [n]$, every $S \subset [n]$ with $i \notin S$, and every pair of assignments $\tau,\sigma: S \rightarrow \{\text{in}, \text{out}\}$, we have
\begin{align*}
    \abs{R_{i}^{\tau} - R_{i}^{\sigma}} \leq C \cdot \alpha^{d(i,S)}
\end{align*}
\end{definition}
\begin{definition}[Strong Spatial Mixing]\label{def:ssm}
Fix a metric $d:[n] \times [n] \rightarrow \R_{\geq0}$. We say a distribution $\mu$ on subsets of $[n]$ exhibits strong spatial mixing w.r.t. $d$ with rate $0 < \alpha < 1$ and constant $C > 0$ if for every $i \in [n]$, every $S \subset [n]$ with $i \notin S$, and every pair of assignments $\tau,\sigma: S \rightarrow \{\text{in}, \text{out}\}$, we have
\begin{align*}
    \abs{R_{i}^{\tau} - R_{i}^{\sigma}} \leq C \cdot \alpha^{d(i,S(\tau,\sigma))}
\end{align*}
where $S(\tau,\sigma) \subset S$ is the subset of elements in $S$ on which the assignments $\tau,\sigma$ differ.
\end{definition}
In the case of a distribution $\mu$ on configurations $\sigma:V \rightarrow \{0,1\}$ on a graph $G=(V,E)$ coming from the hardcore model with activity $\lambda$, our ground set will consist of the vertices $V$. The most natural metric to take is the shortest path metric on the graph $G$. This is what we will do throughout this paper. However, we note there are many works which use different metrics \cite{LLY13, RSTVY13}. We also note there are alternative forms of correlation decay based on computation trees that have been successfully used to obtain approximation algorithms \cite{BGKNT07, GK07}.

It has been known since the work of Kelly \cite{Kel85} that for the hardcore model, weak spatial mixing on the infinite $\Delta$-regular tree holds exactly when $\lambda < \frac{(\Delta-1)^{\Delta-1}}{(\Delta-2)^{\Delta}} \overset{\defin}{=} \lambda_{c}(\Delta)$. Here, $\lambda_{c}(\Delta)$ is known as the critical threshold for the hardcore model on graphs of maximum degree $\leq\Delta$. These results have been subsequently extended to all antiferromagnetic two-state spin systems \cite{LLY12, LLY13, SST14}.

The way the threshold $\lambda_{c}(\Delta)$ is derived is by analyzing when $\abs{f_{\Delta-1}'(\hat{R}_{\Delta-1})}$ is less than $1$. It turns out the gap between $\abs{f_{\Delta-1}'(\hat{R}_{\Delta-1})}$ and $1$ governs the rate $\alpha$ in the definition of weak spatial mixing. \cite{LLY13} quantified this in the following definition.
\begin{definition}[Up-to-$\Delta$ Uniqueness \cite{LLY13}]\label{def:uptodeltaunique}
We say the hardcore model with parameter $\lambda$ is up-to-$\Delta$ unique if for every $1 \leq d < \Delta$, we have $\abs{f_{d}'(\hat{R}_{d})} < 1$, where $\hat{R}_{d}$ denotes the unique fixed point of $f_{d}$. Furthermore, we say $\lambda$ is up-to-$\Delta$ unique with gap $0 < \delta < 1$ if $\abs{f_{d}'(\hat{R}_{d})} \leq 1 - \delta$ for every $1 \leq d < \Delta$.
\end{definition}
It is not hard to show that up-to-$\Delta$ uniqueness with gap $0 < \delta < 1$ is equivalent to $\lambda \leq (1 - \Theta(\delta)) \cdot \lambda_{c}(\Delta)$. Hence, throughout the paper whenever one encounters the phrase ``up-to-$\Delta$ unique with gap $0 < \delta < 1$'', one may safely assume $\lambda \leq (1 - \Theta(\delta))\lambda_{c}(\Delta)$. The rigorous statement of this equivalence and its proof are contained in \cref{sec:gappeduptodeltaunique}.

Surprisingly, Weitz \cite{Wei06} managed to show that for the hardcore model, weak spatial mixing actually implies strong spatial mixing with the same rate $\alpha$, albeit with a worse constant $C$. This was extended in \cite{LLY13, SST14} to all antiferromagnetic two-state spin systems. The way this was done was to first reduce spatial mixing on a general graph to spatial mixing on an associated tree, which we describe now.

\begin{definition}[Self-Avoiding Walk Tree; \cite{SS05}, \cite{Wei06}]
Fix a graph $G= (V,E)$ and a vertex $r \in V$. A self-avoiding walk of length $\ell$ in $G$ beginning at $r$ is a sequence of vertices $r = v_{0},\dots,v_{\ell}$ such that $v_{0},\dots,v_{\ell}$ are all distinct, and $v_{i} \sim v_{i-1}$ for each $i=1,\dots,\ell$. We construct a tree $T = T_{\SAW}(G,r)$ whose vertices correspond to walks $v_{0},\dots,v_{\ell}$ such that either
\begin{enumerate}
    \item $v_{0},\dots,v_{\ell}$ itself is a self-avoiding walk
    \item $v_{0},\dots,v_{\ell-1}$ is a self-avoiding walk, and $v_{\ell} = v_{i}$ for some $i < \ell - 2$, i.e. $v_{\ell}$ closes a cycle; note that $v_{\ell} = v_{\ell-2}$ (backtracking) and $v_{\ell} = v_{\ell-1}$ (staying) are both prohibited
\end{enumerate}
Two such walks are adjacent in $T$ if and only if one extends the other. Let $C(v)$ denote the set of copies of $v$ in $T = T_{\SAW}(G,r)$, i.e. vertices in $T$ corresponding to walks which end at $v$. For vertex $v \in G$, we will let $\ell(v) = \min\{\ell(u) : u \in C(v)\}$ denote the highest level of any copy of $v$ in $T$.

We will further need a boundary condition $\tau_{\SAW}$ on the leaves of $T$. Specifically, for each vertex $v \in G$, we first order its neighbors arbitrarily. Now consider a walk $v_{0},\dots,v_{\ell}$ such that $v_{\ell}$ closes a cycle. Let $i < \ell$ be such that $v_{\ell} = v_{i}$. We assign the vertex in $T$ corresponding to the walk $v_{0},\dots,v_{\ell}$ a spin value of
\begin{enumerate}
    \item $0$ if the neighbor $v_{i+1}$ of $v_{i}$ is larger than the neighbor $v_{\ell-1}$.
    \item $1$ if the neighbor $v_{i+1}$ of $v_{i}$ is larger than the neighbor $v_{\ell-1}$.
\end{enumerate}
Note that $T$ is a finite tree since any vertex in a self-avoiding walk can be visited at most once.
\end{definition}
\begin{figure}
    \centering
    \begin{subfigure}[t]{.5\textwidth}
        \centering
        \begin{tikzpicture}[thick,main node/.style={circle,draw,font=\sffamily\normalsize\itshape\bfseries, inner sep=0pt, minimum size=0.75cm}]
                    
            \node[main node, fill=orange!60] (a) {a};
            \node[main node] (b) [below left=1cm and 1cm of a] {b};
            \node[main node] (c) [below=0.8cm of a] {c};
            \node[main node] (d) [below right=1cm and 1cm of a] {d};
            \node[main node] (e) [below left=1cm and 0.4cm of d] {e};
            \node[main node] (f) [below right=1cm and 0.4cm of d] {f};
        
            \path[every node/.style={font=\sffamily\small}]
                (a) edge node {} (b)
                    edge node {} (c)
                    edge node {} (d)
                (c) edge node {} (d)
                (d) edge node {} (e)
                    edge node {} (f)
                (e) edge node {} (f);
        \end{tikzpicture}
        \caption{Base graph $G$}
        \label{subfig:SAWexamplebase}
    \end{subfigure}%
    \begin{subfigure}[t]{.5\textwidth}
        \centering
        \begin{tikzpicture}[thick,main node/.style={circle,draw,font=\sffamily\normalsize\itshape\bfseries, inner sep=0pt, minimum size=0.5cm}]
                            
            \node[main node, fill=orange!60] (a) {a};
            \node[main node] (ab) [below left=0.4cm and 1cm of a] {b};
            \node[main node] (ac) [below=0.3cm of a] {c};
            \node[main node] (acd) [below left=0.4cm and 1cm of ac] {d};
            \node[main node, fill=red!60] (acda) [below left=0.4cm and 0.6cm of acd] {a};
            \node[main node] (acde) [below=0.3cm of acd] {e};
            \node[main node] (acdef) [below=0.3cm of acde] {f};
            \node[main node, fill=red!60] (acdefd) [below=0.3cm of acdef] {d};
            \node[main node] (acdf) [below right=0.4cm and 0.6cm of acd] {f};
            \node[main node] (acdfe) [below=0.3cm of acdf] {e};
            \node[main node, fill=blue!40] (acdfed) [below=0.3cm of acdfe] {d};
            \node[main node] (ad) [below right=0.4cm and 1cm of a] {d};
            \node[main node] (adc) [below left=0.4cm and 0.6cm of ad] {c};
            \node[main node, fill=blue!40] (adca) [below=0.3cm of adc] {a};
            \node[main node] (ade) [below=0.3cm of ad] {e};
            \node[main node] (adef) [below=0.3cm of ade] {f};
            \node[main node, fill=red!60] (adefd) [below=0.3cm of adef] {d};
            \node[main node] (adf) [below right=0.4cm and 0.6cm of ad] {f};
            \node[main node] (adfe) [below=0.3cm of adf] {e};
            \node[main node, fill=blue!40] (adfed) [below=0.3cm of adfe] {d};
        
            \path[every node/.style={font=\sffamily\small}]
                (a) edge node {} (ab)
                    edge node {} (ac)
                    edge node {} (ad)
                (ac) edge node {} (acd)
                (acd) edge node {} (acda)
                      edge node {} (acde)
                      edge node {} (acdf)
                (acde) edge node {} (acdef)
                (acdef) edge node {} (acdefd)
                (acdf) edge node {} (acdfe)
                (acdfe) edge node {} (acdfed)
                (ad) edge node {} (adc)
                     edge node {} (ade)
                     edge node {} (adf)
                (adc) edge node {} (adca)
                (ade) edge node {} (adef)
                (adef) edge node {} (adefd)
                (adf) edge node {} (adfe)
                (adfe) edge node {} (adfed);
        \end{tikzpicture}
        \caption{$T = T_{\SAW}(G,a)$}
        \label{subfig:SAWexampletree}
    \end{subfigure}%
\caption{Red and blue vertices in $T$ indicate the boundary condition $\tau_{\SAW}$, with red representing ``in'' and blue representing ``out''. These boundary conditions are considered ``structural'' as they depend only upon the cycle structure of the base graph $G$. Here, for each vertex in $G$, we order its neighbors reverse lexicographically.}
\label{fig:SAWexample}
\end{figure}
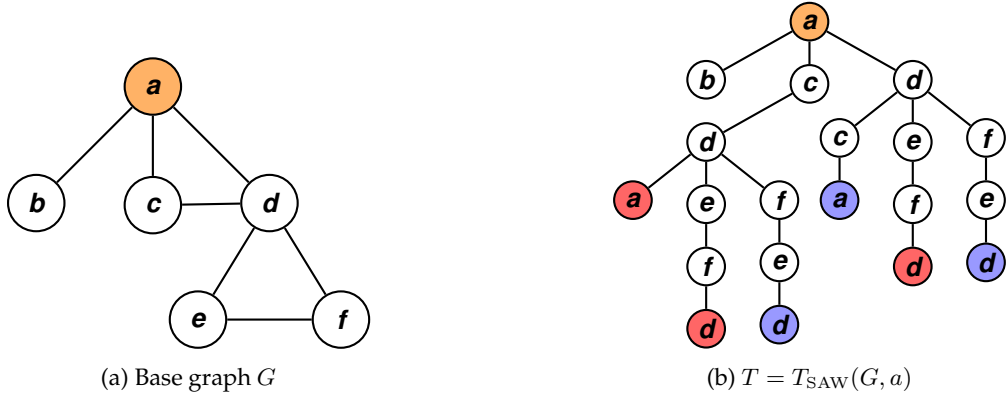
\begin{fact}\label{fact:SAWdistmaxdeg}
Let $G=(V,E)$ be a graph, and let $v \in V$. Let $T = T_{\SAW}(G,r)$.
\begin{enumerate}
    \item The maximum degree of $T$ equals the maximum degree of $G$.
    \item For a vertex $u \in G$, we have $\ell_{T}(u) = d_{G}(u,v)$.
\end{enumerate}
\end{fact}
Note that any partial assignment $\sigma:S \rightarrow \{0,1\}$ in $G$ has a natural extension to a partial assignment $\sigma_{\SAW}$ in $T = T_{\SAW}(G,r)$. Specifically, for every $v \in S$ and every $u \in C(v)$, we have $\sigma_{\SAW}(u) = \sigma(v)$. Fixed vertices in $T$ come in two types. The first comes from $\tau_{\SAW}$, which arise from the cycle structure of $G$. In other words, fixed assignments of this type are ``structural''. The second comes from $\sigma_{\SAW}$, which arise from fixed vertices in $G$. Fixed assignments of this second type are simply ``copied'' assignments. For convenience, whenever we consider a self-avoiding walk tree, we will implicitly assume that $\tau_{\SAW}$ is part of any assignment, without writing it explicitly. In the case of the hardcore model, this is equivalent to simple throwing away all fixed vertices, and neighbors of vertices fixed to ``in'' (i.e. $0$).

One of Weitz's main results is the following \cite{Wei06}.
\begin{theorem}[Theorem 3.1 from \cite{Wei06}]\label{thm:weitzsaw}
Fix a graph $G=(V,E)$ and a vertex $r \in V$, and let $T = T_{\SAW}(G,r)$. Let $\sigma:S \rightarrow \{0,1\}$ be any partial assignment of vertices in $S \subset V$. Then $\Pr_{G}[r \mid \sigma] = \Pr_{T}[r \mid \tau_{\SAW},\sigma_{\SAW}]$ and $R_{G,r}^{\sigma} = R_{T,r}^{\tau_{\SAW},\sigma_{\SAW}}$.
\end{theorem}
Note that \cref{thm:weitzsaw} is generic in that it holds for any distribution $\mu$ described by a two-state spin system. As a consequence, strong spatial mixing holds on $G$ if and only if it holds on $T_{\SAW}(G,r)$ for any $v \in G$. This reduction is advantageous because one can leverage the tree recurrence \cref{eq:treerecurrence} for the ratios of conditional probabilities to prove spatial mixing results.

For the second step, in the case of the hardcore model, \cite{Wei06} showed that weak spatial mixing on the infinite $\Delta$-regular tree implies strong spatial mixing on all trees of maximum degree $\leq\Delta$, and hence, on all graphs of maximum degree $\leq\Delta$. To conveniently state the strong spatial mixing result proved in \cite{Wei06}, we make the following definition.
\begin{definition}
If $T$ is a tree rooted at $r \in T$, we define $R_{T,r}^{\min}(\ell)$ to be the minimum conditional probability ratio that $r$ is assigned $1$ in a random configuration, over all possible marginals of vertices at depth $\ell$ in the subtree rooted at $r$. That is, $R_{T,r}^{\min}(\ell) = \min_{p} R_{T,r}^{p}$, where $p$ is an assignment of marginals of vertices at depth $\ell$ in $T_{u}$. Similarly, define $R_{T,r}^{\max}(\ell)$ to be the maximum such conditional probability ratio. Finally, define $R^{\min}(\ell) = \min_{T,r} R_{T,r}^{\min}(\ell)$ and $R^{\max}(\ell) = \max_{T,r} R_{T,r}^{\max}(\ell)$, where the minimum and maximum are over all trees $T$ rooted at $r$ of maximum degree $\leq\Delta$.
\end{definition}
\begin{remark}
Note that the map $x \mapsto \frac{x}{1 - x}$ is monotone increasing. Hence, the boundary condition $p$ which achieves $R_{T,r}^{\min}(\ell)$ is the one which minimizes $\Pr[r \mid p]$. Essentially, due to the antiferromagnetic nature of the hardcore model, the level-$\ell$ boundary condition minimizing $p_{r}^{\tau}$ is the all-$1$ configuration if $\ell$ is odd, and the all-$0$ configuration if $\ell$ is even. Determining the configuration achieving $R_{r}^{\max}(\ell)$ can be done a similar way.

Note, in particular, that the boundary condition $p$ achieving the minimum or maximum maps the marginals of vertices to $\{0,1\}$; there is no advantage to allowing fractional marginals. However, this formulation of $R_{T,r}^{\min}(\ell)$ and $R_{T,r}^{\max}(\ell)$ will be convenient later.
\end{remark}
\begin{fact}\label{fact:RminRmaxbounds}
We have the inequalities
\begin{enumerate}
    \item $0 = R^{\min}(1) \leq R^{\max}(1) = \lambda$,
    \item $\frac{\lambda}{(1 + \lambda)^{\Delta}} = R^{\min}(2) \leq R^{\max}(2) = \lambda$,
    \item $R^{\min}(\ell) \leq R^{\min}(\ell + 1)$ and $R^{\max}(\ell) \geq R^{\max}(\ell + 1)$ for any $\ell \geq 1$.
\end{enumerate}
\end{fact}
\begin{theorem}[Weak Spatial Mixing Implies Strong Spatial Mixing; \cite{Wei06}]\label{thm:ssm}
Assume $\lambda = (1-\delta)\lambda_{c}(\Delta)$ for some $0 < \delta < 1$. Then there exist constants $C > 0$ and $0 < \alpha < 1$ such that for every tree $T$ of maximum degree $\leq \Delta$ rooted at some $r \in T$, and every level $\ell$, we have the bound
\begin{align*}
    |R_{T,r}^{\min}(\ell) - R_{T,r}^{\max}(\ell)| \leq C \cdot \alpha^{\ell}
\end{align*}
\end{theorem}
Later on in the paper, we will need more precise control over $C,\alpha$. However, the above result is sufficient for the present discussion

\section{The Eigenvalues of the Pairwise Influence Matrix}\label{sec:correlationspecradius}
Our goal in this section is to prove \cref{thm:correlationspecradius}. In fact, we completely characterize the spectrum of $P_{\emptyset}$ in terms of the spectrum of $\Psi_{\mu}$, which immediately implies \cref{thm:correlationspecradius}.
\begin{theorem}\label{thm:correlationeigenvaluesubset}
The spectrum of $P_{\emptyset}$ (as a multiset) is precisely the union of the spectrum of $\frac{1}{n-1}\Psi_{\mu}$ (as a multiset), $n-1$ copies of $-\frac{1}{n-1}$, and an eigenvalue of $1$ (corresponding to the top eigenvalue of $P_{\emptyset}$).
\end{theorem}
Note that this also immediately implies $\Psi_{\mu}$ has real eigenvalues, since $P_{\emptyset}$ has real eigenvalues. The rest of the section is devoted to proving \cref{thm:correlationeigenvaluesubset}.

The main idea behind the proof is to relate the spectra of $P_{\emptyset}$ and of $\frac{1}{n-1}\Psi_{\mu}$ to an intermediate matrix $M_{\emptyset}$. This matrix $M_{\emptyset}$ will be built from $P_{\emptyset}$ by leveraging knowledge of the ``trivial'' eigenvalues and eigenvectors induced purely by the $n$-partite structure of $X^{\mu}$. Towards this, let us first express the entries of $P_{\emptyset}$ in a nice form. Observe that using \cref{eq:faceweights}, we have that for $i,j \in [n]$,
\begin{align*}
    P_{\emptyset}(i,j) = \frac{w(\{i,j\})}{w(\{i\})} \cdot \mathbf{1}[i \neq j] = \frac{1}{n-1} \Pr[j \mid i] \cdot \mathbf{1}[i \neq j]
\end{align*}
Similarly, we have the following for all $i,j \in [n]$.
\begin{align*}
    P_{\emptyset}(i,\overline{j}) &= \frac{1}{n-1} \Pr[\overline{j} \mid i] \cdot \mathbf{1}[i \neq j] \\
    P_{\emptyset}(\overline{i},j) &= \frac{1}{n-1} \Pr[j\mid \overline{i}] \cdot \mathbf{1}[i \neq j] \\
    P_{\emptyset}(\overline{i},\overline{j}) &= \frac{1}{n-1} \Pr[\overline{j} \mid \overline{i}] \cdot \mathbf{1}[i \neq j]
\end{align*}

Now, we compute the stationary distribution of $P_{\emptyset}$. Define $\pi \in \R^{2n}$ entrywise by $\pi(i) = \frac{1}{n}\Pr[i]$ and $\pi(\overline{i}) = \frac{1}{n}\Pr[\overline{i}]$. Observe that $P_{\emptyset}$ is reversible w.r.t. $\pi$. For instance, we have
\begin{align*}
    \pi(i)P_{\emptyset}(i,j) = \frac{\Pr[i] \cdot \Pr[j \mid i]}{n \cdot (n-1)} = \frac{\Pr[i,j]}{n \cdot (n-1)} = \frac{\Pr[j] \cdot \Pr[i \mid j]}{n \cdot (n-1)} = \pi(j)P_{\emptyset}(j,i)
\end{align*}
Hence, $\pi$ is indeed stationary w.r.t. $P_{\emptyset}$. For each element $i \in [n]$, define the vectors $\mathbf{1}^{i},\pi^{i} \in \R^{2n}$ by $\mathbf{1}^{i} = e_{i} + e_{\overline{i}}$ and $\pi^{i} = \pi(i) \cdot e_{i} + \pi(\overline{i}) \cdot e_{\overline{i}}$. In particular, for each $i\in [n]$, $\mathbf{1}^{i},\pi^{i}$ are vectors which are supported on the two entries corresponding to the two different possible assignments of $i$. We now define our intermediate matrix as
\begin{align*}
    M_{\emptyset} = P_{\emptyset} - \frac{n}{n-1} \mathbf{1}\pi^{\top} + \frac{n}{n-1} \sum_{i=1}^{n} \mathbf{1}^{i} (\pi^{i})^{\top}
\end{align*}
We prove the following two claims.
\begin{claim}[Relating $P_{\emptyset}$ and $M_{\emptyset}$]\label{claim:PemptysetMemptyset}
The matrix $P_{\emptyset}$ has eigenvalue $1$ with multiplicity (at least) $1$, and eigenvalue $-\frac{1}{n-1}$ with multiplicity (at least) $n-1$. These are the ``trivial'' eigenvalues of $P_{\emptyset}$. Furthermore, the spectrum of $M_{\emptyset}$ (as a multiset) is precisely the spectrum of $P_{\emptyset}$ with all trivial eigenvalues replaced by $n$ copies of $0$.
\end{claim}
\begin{claim}[Relating $M_{\emptyset}$ and $\frac{1}{n-1}\Psi_{\mu}$]\label{claim:MemptysetPsi}
The spectrum of $M_{\emptyset}$ (as a multiset) is precisely the union of the spectrum of $\frac{1}{n-1}\Psi_{\mu}$ (as a multiset) with $n$ additional copies of $0$.
\end{claim}
\cref{thm:correlationeigenvaluesubset} then follows as an immediate consequence of these two claims, which we now prove. The main idea behind \cref{claim:PemptysetMemptyset} is that the vectors $\mathbf{1}^{i}$ form an orthogonal basis for the span of the right eigenvectors of $P_{\emptyset}$ corresponding to the eigenvalues $1$ and $-\frac{1}{n-1}$ (while the $\pi^{i}$ form an orthogonal basis of the corresponding left eigenvectors). This was actually already observed in \cite{Opp18}. One should view $M_{\emptyset}$ is being defined in a way to ``zero out'' those eigenvalues. For \cref{claim:MemptysetPsi}, the intuition is that $\Psi_{\mu}$ may be obtained from $M_{\emptyset}$ via orthogonal projection.
\begin{proof}[Proof of \cref{claim:PemptysetMemptyset}]
The strategy here is to use another intermediate matrix $Q_{\emptyset} = P_{\emptyset} - \frac{n}{n-1}\mathbf{1}\pi^{\top}$ as a stepping stone for the comparison. Roughly speaking, the shift by $-\frac{n}{n-1}\mathbf{1}\pi^{\top}$ pushes the top eigenvalue down to $-\frac{1}{n-1}$. In particular, $\mathbf{1}$ joins the eigenspace of the eigenvalue $-\frac{1}{n-1}$, allowing for us to compute a nice collection of eigenvectors, namely the $\mathbf{1}^{i}$.

Recall that since $P_{\emptyset}$ is reversible w.r.t. $\pi$, $P_{\emptyset}$ is self-adjoint w.r.t. the inner product
\begin{align*}
    \langle \phi, \psi \rangle_{\pi} = \sum_{i=1}^{n} \wrapp{\pi(i) \cdot \phi(i) \cdot \psi(i) + \pi(\overline{i}) \cdot \phi(\overline{i}) \cdot \psi(\overline{i})}
\end{align*}
By the Spectral Theorem for self-adjoint operators, $\R^{2n}$ admits a basis of eigenvectors $\{v_{k}\}_{k=1}^{2n}$ of $P_{\emptyset}$ which are orthogonal w.r.t. $\langle\cdot,\cdot\rangle_{\pi}$. Without loss, we assume $P_{\emptyset}v_{k} = \lambda_{k}(P_{\emptyset}) \cdot v_{k}$, where $1 = \lambda_{1}(P_{\emptyset}) \geq \dots \geq \lambda_{2n}(P_{\emptyset})$, and $v_{1} = \mathbf{1}$. It follows that
\begin{align*}
    Q_{\emptyset}v_{k} = P_{\emptyset}v_{k} - \frac{n}{n-1}\langle v_{k}, \pi \rangle \cdot \mathbf{1} = \begin{cases}
        \lambda_{k}(P_{\emptyset}) \cdot v_{k}, &\quad\text{if } k \neq 1 \text{ since } \langle v_{k}, \pi \rangle = \langle v_{k}, \mathbf{1} \rangle_{\pi} = 0 \\
        -\frac{1}{n-1} \cdot \mathbf{1}, &\quad\text{if } k = 1 \text{ since } \langle \mathbf{1}, \pi \rangle = \langle \mathbf{1}, \mathbf{1} \rangle_{\pi} = 1
    \end{cases}
\end{align*}
In particular, $Q_{\emptyset}$ has the same spectrum as $P_{\emptyset}$, except with one copy of the eigenvalue $1$ shifted down to $-\frac{1}{n-1}$. Note they also have the same right eigenvectors.

Now we compare $Q_{\emptyset}$ and $M_{\emptyset}$. First, we claim that $Q_{\emptyset}$ is also self-adjoint w.r.t. $\langle\cdot,\cdot\rangle_{\pi}$ since $P_{\emptyset}$ and $\mathbf{1}\pi^{\top}$ are. Hence, $\R^{2n}$ also admits a basis of eigenvectors $\{u_{k}\}_{k=1}^{2n}$ of $Q_{\emptyset}$ which are orthogonal w.r.t. $\langle\cdot,\cdot\rangle_{\pi}$, where we assume $Q_{\emptyset}u_{k} = \lambda_{k}(Q_{\emptyset}) \cdot u_{k}$. Now, we explicitly write down $n$ of these eigenvectors corresponding to the eigenvalue $-\frac{1}{n-1}$. We claim that $\mathbf{1}^{i}$ is a right eigenvector of $Q_{\emptyset}$ with eigenvalue $-\frac{1}{n-1}$ for each $i=1,\dots,n$. For this we compute that
\begin{align*}
    Q_{\emptyset}\mathbf{1}^{i} = P_{\emptyset}\mathbf{1}^{i} - \frac{n}{n-1} \langle \mathbf{1}^{i}, \pi \rangle \cdot \mathbf{1} \underset{(*)}{\underbrace{=}} \frac{1}{n-1} (\mathbf{1} - \mathbf{1}^{i}) - \frac{1}{n-1} \mathbf{1}^{i} = -\frac{1}{n-1} \mathbf{1}^{i}
\end{align*}
To justify $(*)$, we use that $\langle \mathbf{1}^{i}, \pi \rangle = \frac{1}{n} (\Pr[i] + \Pr[\overline{i}]) = \frac{1}{n}$, and for all $j\in [n]$,
\begin{align*}
    (P_{\emptyset}\mathbf{1}^{i})(j) &= P_{\emptyset}(j,i) + P_{\emptyset}(j,\overline{i}) = \frac{1}{n-1} \wrapp{\Pr[i \mid j] + \Pr[\overline{i} \mid j]} \cdot \mathbf{1}[i \neq j] = \frac{1}{n-1} \mathbf{1}[i \neq j] \\
    (P_{\emptyset}\mathbf{1}^{i})(\overline{j}) &= P_{\emptyset}(\overline{j},i) + P_{\emptyset}(\overline{j},\overline{i}) = \frac{1}{n-1} \wrapp{\Pr[i \mid \overline{j}] + \Pr[\overline{i} \mid \overline{j}]} \cdot \mathbf{1}[i \neq j] = \frac{1}{n-1} \mathbf{1}[i \neq j]
\end{align*}

With this, we have that since the $\mathbf{1}^{i}$ are mutually orthogonal eigenvectors, $\R^{2n}$ admits a basis $\{u_{k}\}_{k=1}^{2n}$, containing the $\mathbf{1}^{i}$, of eigenvectors of $Q_{\emptyset}$ which are mutually orthogonal w.r.t. $\langle\cdot,\cdot\rangle_{\pi}$. Now, we have that
\begin{align*}
    M_{\emptyset}u_{k} &= Q_{\emptyset}u_{k} + \frac{n}{n-1} \sum_{i=1}^{n} \langle u_{k}, \pi^{i} \rangle \cdot \mathbf{1}^{i} \\
    &= \lambda_{k}(Q_{\emptyset})u_{k} + \frac{n}{n-1} \sum_{i=1}^{n} \langle u_{k}, \mathbf{1}^{i} \rangle_{\pi} \cdot \mathbf{1}^{i} = \begin{cases}
        0, &\quad\text{if } u_{k} = \mathbf{1}^{i} \text{ for some } i \\
        \lambda_{k}(Q_{\emptyset}) \cdot u_{k}, &\quad\text{o.w.}
    \end{cases}
\end{align*}
where in the final step, we use that $\langle \mathbf{1}^{i}, \mathbf{1}^{i} \rangle_{\pi} = \pi(i) + \pi(\overline{i}) = \frac{1}{n}(\Pr[i] + \Pr[\overline{i}]) = \frac{1}{n}$ for the case $u_{k} = \mathbf{1}^{i}$. The claim follows.
\end{proof}
\begin{proof}[Proof of \cref{claim:MemptysetPsi}]
First, we observe that $M_{\emptyset}$ has the following convenient block structure:
\begin{align*}
    M_{\emptyset} = \begin{bmatrix} A & -A \\ B & -B \end{bmatrix}
\end{align*}
where $A,B \in \R^{n \times n}$ are matrices with entries
\begin{align*}
    A(i,j) = \frac{1}{n-1} \wrapp{\Pr[j \mid i] - \Pr[j]} \cdot \mathbf{1}[i \neq j] \quad\quad
    B(i,j) = \frac{1}{n-1} \wrapp{\Pr[j \mid \overline{i}] - \Pr[j]} \cdot \mathbf{1}[i \neq j]
\end{align*}
Furthermore, it is straightforward to see that $\frac{1}{n-1} \Psi_{\mu} = A - B$. If we assume the truth of this observation, then the claim follows, since we have that the characteristic polynomial of $M_{\emptyset}$ is given by
\begin{align*}
    \det(xI - M_{\emptyset}) &= \det\begin{bmatrix} xI - A & A \\ -B & xI + B \end{bmatrix} \\
    &= \det((xI - A)(xI + B) + AB) \tag*{(for instance, using \cite[Theorem 3]{Sil00})}\\ %http://www.ee.iisc.ac.in/people/faculty/prasantg/downloads/blocks.pdf
    &= \det(x^{2}I - xA + xB) \\
    &= x^{n} \det(xI - (A - B))
\end{align*}
It remains to justify the observation that $M_{\emptyset}$ may be expressed as such as block matrix. A calculation of the entries of $M_{\emptyset}$ reveals that for all $i,j \in [n]$, we have
\begin{align*}
    M_{\emptyset}(i,j) &= P_{\emptyset}(i,j) - \frac{n}{n-1}\underset{\frac{1}{n} \Pr[j]}{\underbrace{(\mathbf{1}\pi^{\top})(i,j)}} + \frac{n}{n-1} \sum_{k=1}^{n} \underset{\frac{1}{n}\Pr[j] \cdot \mathbf{1}[i = j = k]}{\underbrace{(\mathbf{1}^{k}(\pi^{k})^{\top})(i,j)}} \\
    &= \frac{1}{n-1} \Pr[j \mid i] \cdot \mathbf{1}[i \neq j] - \frac{1}{n-1} \Pr[j] + \frac{1}{n-1} \Pr[j] \cdot \underset{1 - \mathbf{1}[i \neq j]}{\underbrace{\mathbf{1}[i = j]}} = A(i,j)
\end{align*}
A nearly identical calculation also yields
\begin{align*}
    M_{\emptyset}(i,\overline{j}) = -A(i,j) \quad\quad\quad M_{\emptyset}(\overline{i},j) = B(i,j) \quad\quad\quad M_{\emptyset}(\overline{i},\overline{j}) = -B(i,j)
\end{align*}
from which, the observation on the block structure of $M_{\emptyset}$ immediately follows.
\end{proof}
\begin{remark}\label{rem:partitegeneralization}
The essence of the proof lies in the fact that there are ``trivial'' eigenvectors whose structure and corresponding eigenvalue derive purely from the fact that in the weighted graph with vertex set $\{i,\overline{i} : i \in [n]\}$ corresponding to $P_{\emptyset}$, there are is no edge between the vertices $i$ and $\overline{i}$, for all $i \in [n]$. This is a generalization of the fact that the random walk matrix of any weighted bipartite graph always has eigenvalue $-1$, purely due to bipartiteness. These observations generalize in a straightforward fashion to all partite complexes in the sense that for any $d$-dimensional $d$-partite weighted complex $(X,w)$ with parts $U_{1},\dots,U_{d}$, the indicator vectors $\mathbf{1}^{U_{1}},\dots,\mathbf{1}^{U_{d}}$ are eigenvectors of $P_{\emptyset} - \frac{d}{d-1}\mathbf{1}\pi^{\top}$ with eigenvalue $-\frac{1}{d-1}$. This was also observed in \cite{Opp18}.
\end{remark}

\section{Influence Decoupling in Weitz's Self-Avoiding Walk Tree}\label{sec:decoupling}
In this section, we take a step towards proving \cref{thm:hardcoretotalinfbound}. Specifically, we focus on bounding
\begin{align*}
    \sum_{u \in V : u \neq v} \abs{\Psi_{\mu}(u,v)}
\end{align*}
where from now on, we take $\mu$ to be the distribution corresponding to the hardcore distribution on input graph $G=(V,E)$ with parameter $\lambda > 0$. Here, the relevant uniqueness threshold is given by $\lambda_{c}(\Delta) = \frac{(\Delta-1)^{\Delta-1}}{(\Delta-2)^{\Delta}}$.

Before we proceed to bound this quantity for general graphs, we note that one can easily deduce an $O(1)$ upper bound for amenable graphs (i.e. graphs such that the balls around any vertex grows subexponentially fast in the radius) in a black-box fashion directly using strong spatial mixing \cref{def:ssm}, thus recovering some of the previously known connections between spatial mixing properties of the hardcore distribution, and temporal mixing of the Glauber dynamics \cite{DSVW02, Wei04}. This class of graphs notably includes lattices such as $\Z^{d}$, but exclude most graphs such as expanders. Hence, instead of applying strong spatial mixing as a black-box, we revisit its proof, and modify it as necessary.

The high-level strategy is to convert this problem on general graphs to bounding a similar quantity for trees. We do this by leveraging the self-avoiding walk tree construction of \cite{Wei06}. However, since a vertex $u \in G$ may have many copies in the corresponding self-avoiding walk tree $T = T_{\SAW}(G,r)$, we need to ``decouple'' these copies so as to obtain single-vertex influences again.
\begin{definition}[$R$-Pseudoinfluence]
Recall that for a fixed tree $T$ rooted at $r$ with boundary condition $p:A \rightarrow [0,1]$ (where $A$ is a subset of vertices not containing $r$), we write $R_{T,r}^{p} = \frac{\Pr[r \mid p]}{1 - \Pr[r \mid p]}$. For a vertex $v \in T$ with $v \neq r$, we define the $R$-pseudoinfluence of $v$ on the root $r$ by the quantity
\begin{align*}
    \mathcal{R}_{T,r}^{v} = \max_{p} \mathcal{R}_{T,r}^{v,p} \quad\quad\text{where}\quad\quad \mathcal{R}_{T,r}^{v,p} = \abs{R_{T,r}^{v^{0},p} - R_{T,r}^{v^{1},p}}
\end{align*}
and the maximum is taken over all partial assignments $p:L_{r}(\ell(v)) \smallsetminus \{v\} \rightarrow [0,1]$ of marginal values. Again, we drop the subscript $T$ when the tree is clear from context.
\end{definition}
\begin{remark}\label{rem:assm}
It was pointed out to us by Zongchen Chen and Eric Vigoda that our notion of $R$-pseudoinfluence is very related to the notion of ``aggregate strong spatial mixing'' used in \cite{MS13} to analyze the Glauber dynamics, and in \cite{BCV20} to analyze the Swendsen-Wang dynamics, both for the ferromagnetic Ising model. In fact, it turns out our result also directly implies aggregate strong spatial mixing for arbitrary trees in the uniqueness regime $\lambda < \lambda_{c}(\Delta)$.
\end{remark}
Our first step is to do the decoupling using the $R$-pseudoinfluence. The second step is to bound the total $R$-pseudoinfluence of vertices in a tree on the root. These steps are captured in the following two results. We emphasize \cref{lem:decoupling} is generic, and holds for any two-state spin system.
\begin{lemma}[Decoupling]\label{lem:decoupling}
Consider the hardcore distribution $\mu$ on a graph $G=(V,E)$ with parameter $\lambda > 0$. Fix a vertex $r \in G$ and let $T = T_{\SAW}(G,r)$. Then the following inequality holds:
\begin{align*}
    \sum_{v \in G : v \neq r} \abs{\Psi_{\mu}(v,r)} \leq 2\sum_{v \in T : v \neq r} \mathcal{R}_{r}^{v}
\end{align*}
\end{lemma}
In particular, to bound $\sum_{v \in G : v \neq r} \abs{\Psi_{\mu}(v,r)}$, it suffices to bound $\sum_{v \in T : v \neq r} \mathcal{R}_{r}^{v}$ for every tree $T$ of maximum degree $\leq \Delta$ rooted at $r$. This motivates the next result.
\begin{proposition}[$R$-Pseudoinfluence Bound]\label{prop:betazeromodinfbound}
Assume $\lambda$ is up-to-$\Delta$ unique with gap $0 < \delta < 1$. Then for every tree $T$ of maximum degree $\leq \Delta$ rooted at $r$, we have the bound
\begin{align*}
    \sum_{v \in T : v \neq r} \mathcal{R}_{r}^{v} \leq \exp(O(1/\delta))
\end{align*}
\end{proposition}
With these two lemmas in hand, we may prove \cref{thm:hardcoretotalinfbound}.
\begin{proof}[Proof of \cref{thm:hardcoretotalinfbound}]
Let $0 < \delta < 1$ and take $\lambda = (1 - \delta)\lambda_{c}(\Delta)$. By \cref{lem:gappeduptodeltaunique}, $\lambda$ is up-to-$\Delta$ unique. Let $G = (V,E)$ be a graph with maximum degree $\leq \Delta$ and let $T = T_{\SAW}(G,r)$ be the self-avoiding walk tree rooted at an arbitrary vertex $r \in V$. By \cref{lem:decoupling}
\begin{align*}
    \sum_{u \in V : u \neq v} \abs{\Psi_{\mu}(u,v)} \leq 2\sum_{v \in T : v \neq r} \mathcal{R}_{r}^{v}
\end{align*}
By \cref{prop:betazeromodinfbound}, the right-hand side is bounded above by $\exp(O(1/\delta))$ as desired.
\end{proof}

It remains to prove \cref{lem:decoupling} and \cref{prop:betazeromodinfbound}. The rest of the section is devoted to proving the former. The proof of the latter is contained in the following section.
\begin{proof}[Proof of \cref{lem:decoupling}]
Recall that for a vertex $v \in G$ with $v \neq r$, $C(v)$ denotes the set of all copies of $v$ in $T = T_{\SAW}(G,r)$. In particular, $\{C(v) : v \in G, v \neq r\}$ is a partition of the vertices of $T$ (excluding the root $r$ of $T$). Hence, to prove the claim, it suffices to show that for each $v \in G$ with $v \neq r$, we have
\begin{align*}
    \abs{\Psi_{\mu}(v,r)} \leq 2\sum_{u \in C(v)} \mathcal{R}_{r}^{u}
\end{align*}
Towards this, we define an intermediate quantity which we use only for the purposes of this proof. Specifically, we define the pseudoinfluence as
\begin{align*}
    \mathcal{I}_{r}^{u} = \max_{p} \abs{\Pr\wrapb{r \mid u^{0},p} - \Pr\wrapb{r \mid u^{1},p}}
\end{align*}
where again the maximum is taken over all partial assignments $p:L_{r}(\ell(u)) \smallsetminus \{u\} \rightarrow [0,1]$. Observe that $\mathcal{I}_{r}^{u} \leq \mathcal{R}_{r}^{u}$ trivially since
\begin{align*}
    \abs{R_{r}^{u^{0},p} - R_{r}^{u^{1},p}} &= \abs{\frac{\Pr\wrapb{r \mid u^{0},p}}{1 - \Pr\wrapb{r \mid u^{0},p}} - \frac{\Pr\wrapb{r \mid u^{1},p}}{1 - \Pr\wrapb{r \mid u^{1},p}}} \\
    &= \frac{\abs{\Pr\wrapb{r \mid u^{0},p} - \Pr\wrapb{r \mid u^{1},p}}}{(1 - \Pr\wrapb{r \mid u^{0},p})(1 - \Pr\wrapb{r \mid u^{1},p})} \geq \abs{\Pr\wrapb{r \mid u^{0},p} - \Pr\wrapb{r \mid u^{1},p}}
\end{align*}
holds for any $p$. Hence, it suffices to prove
\begin{align*}
    \abs{\Psi_{\mu}(v,r)} \leq 2\sum_{u \in C(v)} \mathcal{I}_{r}^{u}
\end{align*}
We prove a more general claim.
\begin{claim}\label{claim:generaldecoupling}
Fix a set of vertices $A$ in $T$ such that no vertex of $A$ has an ancestor also in $A$. Then
\begin{align*}
    \abs{\Pr[r \mid \sigma_{A,0}] - \Pr[r \mid \sigma_{A,1}]} \leq 2\sum_{v \in A} \mathcal{I}_{r}^{v}
\end{align*}
where $\sigma_{A,0}$ assigns all vertices in $A$ to $0$, and $\sigma_{A,1}$ assigns all vertices in $A$ to $1$.
\end{claim}
Taking $A$ to be the set of vertices in $C(v)$ such that no ancestor is also in $C(v)$, we obtain our main decoupling result. All that remains is to prove the claim.
\begin{proof}[Proof of \cref{claim:generaldecoupling}]
We go by induction on the size of $A$. The base case is obvious. Let $v$ denote the closest vertex of $A$ to $r$. More specifically, let $v \in A$ be such that $d(r,v) = d(r,A)$. Let $B = A \smallsetminus \{v\}$ For the inductive step, observe by Triangle Inequality that
\begin{align*}
    \abs{\Pr[r \mid \sigma_{A,0}] - \Pr[r \mid \sigma_{A,1}]} &\leq \abs{\Pr[r \mid \sigma_{B,0},v^{0}] - \Pr[r \mid \sigma_{B,0}]} \\
    &+ \abs{\Pr[r \mid \sigma_{B,0}] - \Pr[r \mid \sigma_{B,1}]} \\
    &+ \abs{\Pr[r \mid \sigma_{B,1}] - \Pr[r \mid \sigma_{B,1},v^{1}]}
\end{align*}
By the inductive hypothesis, $\abs{\Pr[r \mid \sigma_{B,0}] - \Pr[r \mid \sigma_{B,1}]} \leq 2 \sum_{u \in B} \mathcal{I}_{r}^{u} = 2\sum_{u \in A : u \neq v} \mathcal{I}_{r}^{u}$. Hence, it suffices to bound the first and last terms by $\mathcal{I}_{r}^{v}$. We focus on the first term $\abs{\Pr[r \mid \sigma_{B,0},v^{0}] - \Pr[r \mid \sigma_{B,0}]}$; the same argument works for the last term $\abs{\Pr[r \mid \sigma_{B,1}] - \Pr[r \mid \sigma_{B,1},v^{1}]}$. Since the tree recursion is a monotone function in each individual variable, $\Pr[r \mid \sigma_{B,0}]$ is a monotone function in the marginal of $v$. It follows that
\begin{align*}
    \abs{\Pr[r \mid \sigma_{B,0},v^{0}] - \Pr[r \mid \sigma_{B,0}]} \leq \abs{\Pr[r \mid \sigma_{B,0},v^{0}] - \Pr[r \mid \sigma_{B,0},v^{1}]}
\end{align*}
But
\begin{align*}
    \abs{\Pr[r \mid \sigma_{B,0},v^{0}] - \Pr[r \mid \sigma_{B,0},v^{1}]} \leq \mathcal{I}_{r}^{v}
\end{align*}
clearly holds simply because all vertices in $B$ are at least as far from $r$ as $v$ is, and hence, the effect of $B$ may be simulated by enforcing appropriate marginals on all vertices distance $d(r,v)$. The same argument shows that $\abs{\Pr[r \mid \overline{B}] - \Pr[r \mid \overline{B}, \overline{v}]} \leq \mathcal{I}_{v}^{r}$, and so the claim follows.
\end{proof}
\end{proof}

\subsection{\texorpdfstring{$R$-Pseudoinfluence Decay}{RPseudoinfluenceDecay}}
Our goal is now to prove \cref{prop:betazeromodinfbound}. To do this, we write
\begin{align*}
    \sum_{v \in T : v \neq r} \mathcal{R}_{r}^{v} = \sum_{\ell=1}^{\infty} \sum_{v \in L_{r}(\ell)} \mathcal{R}_{r}^{v}
\end{align*}
Thus, it suffices to bound $\sum_{v \in L_{r}(\ell)} \mathcal{R}_{r}^{v}$ for each level $\ell$. We show that this quantity in fact decays exponentially fast as $\ell$ increases when $\lambda < \lambda_{c}(\Delta)$. Specifically, to prove \cref{prop:betazeromodinfbound}, we use the following two lemmas, which precisely quantify the decay rate.
\begin{proposition}[Decay Rate Bound]\label{prop:betazeroratiosdecayrate}
Assume $\lambda$ is up-to-$\Delta$ unique with gap $0 < \delta < 1$. Then there exists $\ell_{0} = \Theta(1/\delta)$ such that for every tree $T$ of maximum degree $\leq \Delta$ rooted at $r$ and any $\ell > \ell_{0}$, we have the bound
\begin{align*}
    \frac{\sum_{v \in L_{r}(\ell)} \mathcal{R}_{r}^{v}}{\max_{u \in L_{r}(\ell - \ell_{0})} \wrapc{\sum_{v \in L_{u}(\ell_{0})} \mathcal{R}_{u}^{v}}} \leq O(1) \cdot \sqrt{1-\delta}^{\ell-\ell_{0}}
\end{align*}
\end{proposition}
We prove \cref{prop:betazeroratiosdecayrate}  in the next section.
Roughly speaking, the reason for the assumption $\ell>\ell_0$ above is that we can exploit spatial mixing to argue that the marginals of the root is independent of the boundary condition at level $\ell$, for a large enough $\ell_0$; see \cref{sec:potentialmethodratiosinfdecay} for more details. For  $\ell < \ell_{0}$ we use the following lemma.
\begin{lemma}[Trivial ``Decay'' Rate]\label{lem:betazerotrivialdecay}
Assume $\lambda$ up-to-$\Delta$ unique with gap $0 < \delta < 1$. Then for any tree $T$ of maximum degree $\leq \Delta$ rooted at $r$ and any $\ell > 0$, we have
\begin{align*}
    \frac{\sum_{v \in L_{r}(\ell)} \mathcal{R}_{r}^{v}}{\max_{u \in L_{r}(1)}\wrapc{\sum_{v \in L_{u}(\ell-1)}\mathcal{R}_{u}^{v}}} \leq O(1)
\end{align*}
Furthermore, for the first level, we have the inequality
\begin{align*}
    \sum_{v \in L_{r}(1)} \mathcal{R}_{r}^{v} \leq O(1)
\end{align*}
\end{lemma}
\begin{proof}%[Proof of \cref{lem:betazerotrivialdecay}]
Since $\lambda \leq \lambda_{c}(\Delta) \leq O(1/\Delta)$, it suffices to show that $\frac{\sum_{v \in L_{r}(\ell)} \mathcal{R}_{r}^{v}}{\max_{u \in L_{r}(1)}\wrapc{\sum_{v \in L_{u}(\ell-1)}\mathcal{R}_{u}^{v}}} \leq (\Delta-1)\lambda$.

We proceed via the Mean Value Theorem. For any fixed $v \in L_{u}(\ell-1)$ where $u \in L_{r}(1)$, and for any $p: L_{r}(\ell(v)) \smallsetminus \{v\} \rightarrow [0,1]$, we have there exists $\tilde{\mathbf{R}}$ such that
\begin{align*}
    \abs{R_{r}^{v^{0},p} - R_{r}^{v^{1},p}} = \abs{\partial_{R_{u}} F(\tilde{\mathbf{R}})} \cdot \abs{R_{u}^{v^{0},p} - R_{u}^{v^{1},p}}
\end{align*}
Since
\begin{align*}
    \abs{\partial_{R_{u}}F(\mathbf{R})} =  \frac{F(\mathbf{R})}{R_{u} + 1}
\end{align*}
is monotone decreasing in each coordinate, we obtain an upper bound of
\begin{align*}
    \abs{R_{r}^{v^{0},p} - R_{r}^{v^{1},p}} \leq \abs{\partial_{R_{u}}F(\mathbf{0})} \cdot \mathcal{R}_{u}^{v} =  \lambda \cdot \mathcal{R}_{u}^{v}
\end{align*}
As this holds for all $p:L_{r}(\ell(v)) \smallsetminus \{v\} \rightarrow [0,1]$, it follows that $\mathcal{R}_{r}^{v} \leq \lambda \cdot \mathcal{R}_{u}^{v}$. Summing over all $v \in L_{r}(\ell)$, we have
\begin{align*}
    \sum_{v \in L_{r}(\ell)} \mathcal{R}_{r}^{v} \leq \lambda \sum_{u \in L_{r}(1)} \sum_{v \in L_{u}(\ell-1)} \mathcal{R}_{u}^{v} \leq (\Delta-1) \lambda \max_{u \in L_{r}(1)} \wrapc{\sum_{v \in L_{u}(\ell-1)} \mathcal{R}_{u}^{v}}
\end{align*}
as desired.

For the bound on $\sum_{v \in L_{r}(1)} \mathcal{R}_{r}^{v}$, note that it suffices to bound $\mathcal{R}_{r}^{v}$ by $\lambda$ for any $v \in L_{r}(1)$. To do this, fix any $p : L_{r}(1) \smallsetminus \{v\} \rightarrow [0,1]$. We have $R_{r}^{v^{1},p} = 0$ simply because conditioning $v$ to be ``in'' forces $r$ to be ``out''. Hence,
\begin{align*}
    \abs{R_{r}^{v^{0},p} - R_{r}^{v^{1},p}} = R_{r}^{v^{0},p}
\end{align*}
Since the tree recursion is a monotone function in each coordinate, $R_{r}^{v^{0},p}$ is maximized when $p$ is identically zero on $L_{r}(1) \smallsetminus \{v\}$, in which case, $R_{r}^{v^{0},p} = \lambda$. It follows that $\mathcal{R}_{r}^{v} \leq \lambda$ as desired.
\end{proof}

These two results together immediately imply \cref{prop:betazeromodinfbound}.
\begin{proof}[Proof of \cref{prop:betazeromodinfbound}]
Using \cref{lem:betazerotrivialdecay}, we have for any $\ell \leq \ell_{0}$ that
\begin{align*}
    \sum_{v \in L_{r}(\ell)} \mathcal{R}_{r}^{v} \leq O(C^{\ell})
\end{align*}
for a universal constant $C > 0$. When $\ell > \ell_{0}$, we have
\begin{align*}
    \sum_{v \in L_{r}(\ell)} \mathcal{R}_{r}^{v} \leq O(1) \cdot C^{\ell_{0}} \cdot \sqrt{1 - \delta}^{\ell-\ell_{0}}
\end{align*}
Hence, summing over all $\ell$ and using $\ell_{0} \leq O(1/\delta)$, we obtain that
\begin{align*}
    \sum_{v \in T : v \neq r} \mathcal{R}_{r}^{v} \leq O(1) \cdot C^{\ell_{0}} \cdot \exp(O(1/\delta)) \cdot \underset{\leq O(1/\delta)}{\underbrace{\sum_{\ell=1}^{\infty} \sqrt{1 - \delta}^{\ell}}} \leq O(1) \cdot \exp(O(1/\delta))
\end{align*}
The claim follows.
\end{proof}

\section{Bounding the \texorpdfstring{$R$-Pseudoinfluence Decay}{RPseudoinfluenceDecay}: The Potential Method}\label{sec:potentialmethodratiosinfdecay}
Our goal in this section is to prove \cref{prop:betazeroratiosdecayrate}.
We use the potential method (otherwise known as the message decay argument), which has been successfully used in \cite{LLY12, LLY13, RSTVY13, SST14, SSSY15} to establish strong spatial mixing all the way up to the uniqueness threshold.

\begin{definition}[$\varphi$-Pseudoinfluence]
We say a function $\varphi:[0,\infty)\mapsto [a,b]$, for $a<b$, with derivative $\Phi = \varphi'$, is a potential function if it is
\begin{enumerate}
    \item continuously differentiable,
    \item strictly monotone increasing, i.e., $\Phi$ is strictly positive,
    \item concave, i.e., $\Phi$ is decreasing.
\end{enumerate}
For a boundary condition $p:A \rightarrow [0,1]$, where $A$ is a subset of vertices not containing $r$, let $K_{r}^{p} = \varphi(R_{r}^{p})$. Again, we define
\begin{align*}
    \mathcal{K}_{r}^{v,p} \overset{\defin}{=} \abs{K_{r}^{v^{0},p} - K_{r}^{v^{1},p}} = \abs{\varphi(R_{r}^{v^{0},p}) - \varphi(R_{r}^{v^{1},p})}
\end{align*}
Define the $\varphi$-pseudoinfluence of a vertex $v$ on $r$ as
\begin{align*}
    \mathcal{K}_{r}^{v} \overset{\defin}{=} \max_{p : L_{r}(\ell(v)) \smallsetminus \{v\} \rightarrow [0,1]} K_{r}^{v,p}
\end{align*}
Finally, we define
\begin{align*}
    K_{r}^{\min}(\ell) &\overset{\defin}{=} \min_{p : L_{r}(\ell) \rightarrow [0,1]} K_{r}^{p} = \varphi(R_{r}^{\min}(\ell)) \\
    K_{r}^{\max}(\ell) &\overset{\defin}{=} \max_{p:L_{r}(\ell) \rightarrow [0,1]} K_{r}^{p} = \varphi(R_{r}^{\max}(\ell))
\end{align*}
\end{definition}
Let us now fix the potential function that we will use. In this work, we use the potential function introduced in \cite{LLY13}. We define $\varphi$ as
\begin{align*}
    \varphi(R) &\overset{\defin}{=} 2\log(\sqrt{R} + \sqrt{R + 1}) \\
    \Phi(R) &\overset{\defin}{=} \varphi'(R) = \frac{1}{\sqrt{R(R + 1)}}
\end{align*}
We note that since $\Phi$ is continuous, positive, and decreasing, we have $\varphi$ is continuously differentiable, strictly monotone increasing and concave as desired. One additional feature of this potential function is that it has no dependence on $\lambda$ or $\Delta$. While it may be comforting to have an explicit expression for $\varphi$, all of our proofs rely at most on the explicit expression for $\Phi$, rather than $\varphi$. For the derivation and further discussion of this potential function, we refer the reader to \cite{LLY13}.

To control $\sum_{v \in L_{r}(\ell)} \mathcal{R}_{r}^{v}$, it turns out it suffices to control the decay of $\sum_{v \in L_{r}(\ell)} \mathcal{K}_{r}^{v}$ as $\ell$ increases, as we will see later.
\begin{proposition}[$\varphi$-Pseudoinfluence Decay Rate Bound]\label{prop:messagedecayrate}
Assume $\lambda$ is up-to-$\Delta$ unique with gap $0 < \delta < 1$ (see \cref{def:uptodeltaunique}). For $\ell \geq 2$, assume that there exists $\eta \leq 1/2$ such that for all $u \in L_{r}(1)$, we have the inequality $\abs{R_{u}^{\min}(\ell-1) - R_{u}^{\max}(\ell-1)} \leq \eta$. Then,
\begin{align*}
    \frac{\sum_{v \in L_{r}(\ell)} \mathcal{K}_{r}^{v}}{\max_{u \in L_{r}(1)}\wrapc{\sum_{v \in L_{u}(\ell-1)} \mathcal{K}_{u}^{v}}} \leq \wrapp{1 + 2\eta}^{\Delta+1} \sqrt{1 - \delta}.
\end{align*}
\end{proposition}
Unfortunately, due to the additional error factor of $\wrapp{1 + 2\eta}^{\Delta+1}$, we must control $\eta = \eta(\ell)$. To do this, we leverage the strong spatial mixing result proved in \cite{Wei06}. We state a ``precise'' version here, where the constant in front of the decay is stated explicitly.
\begin{definition}\label{def:etas}
Define $\eta^{*} = \frac{R^{\max}(2)}{R^{\min}(2)} \cdot \abs{R^{\min}(2) - R^{\max}(2)}$. Note by \cref{fact:RminRmaxbounds} and the fact that $\lambda \leq O(1/\Delta)$, we have 
\begin{align}\label{eq:etastarbound}
    \eta^{*} \leq \frac{\lambda}{\frac{\lambda}{(1 + \lambda)^{\Delta}}} \cdot \abs{\lambda - \frac{\lambda}{(1 + \lambda)^{\Delta}}} \leq O(1/\Delta).
\end{align}
\end{definition}
\begin{proposition}[Strong Spatial Mixing \cite{Wei06}]\label{prop:precisessm}
Assume that $\lambda$ is up-to-$\Delta$ unique with gap $0 < \delta < 1$. Then for all trees $T$ rooted at $r$ of maximum degree $\leq\Delta$, we have
\begin{align*}
    \abs{R_{T,r}^{\min}(\ell) - R_{T,r}^{\max}(\ell)} \leq \sqrt{1 - \delta}^{\ell -2} \cdot \eta^{*}
\end{align*}
\end{proposition}
For the sake of completeness, we provide a proof of this specific bound in \cref{sec:precisessm}.

With these results in hand, we may deduce \cref{prop:betazeroratiosdecayrate}.
\subsection{Proof of \texorpdfstring{\cref{prop:betazeroratiosdecayrate}}{betazeroratiosdecayrate}}
In order to apply \cref{prop:messagedecayrate} and \cref{prop:precisessm}, we must relate $\mathcal{R}_{r}^{v}$ to $\mathcal{K}_{r}^{v}$. This is done in the following lemma.
\begin{lemma}[Relating $R$-Pseudoinfluence to $\varphi$-Pseudoinfluence]\label{lem:messageratiosrelate} 
Let $T$ be a tree rooted at $r$. For any $\ell\geq 1$ and any vertex $v\in L_r(\ell)$, we have the bound
\begin{align*}
    \Phi(R^{\max}(\ell)) \cdot \mathcal{R}_{r}^{v} \leq \mathcal{K}_{r}^{v} \leq \Phi(R^{\min}(\ell)) \cdot \mathcal{R}_{r}^{v}.
\end{align*}
\end{lemma}
\begin{proof}
First, observe by the Mean Value Theorem and monotonicity of $\Phi$ that for any $R_{0} \leq R_{1}$, we have the inequalities
\begin{align*}
    \Phi(R_{1}) \cdot \abs{R_{1} - R_{0}} \leq \abs{\varphi(R_{1}) - \varphi(R_{0})} \leq \Phi(R_{0}) \cdot \abs{R_{1} - R_{0}}.
\end{align*}
Now, fix $v \in L_{r}(\ell)$; we prove the RHS inequality in lemma's statement.  For any boundary condition $p:L_{r}(\ell(v)) \smallsetminus \{v\} \rightarrow [0,1]$, we have
\begin{align*}
    \abs{\varphi(R_{r}^{v^{0},p}) - \varphi(R_{r}^{v^{1},p})} \leq \max\{\Phi(R_{r}^{v^{0},p}), \Phi(R_{r}^{v^{1},p})\} \cdot \abs{R_{r}^{v^{0},p} - R_{r}^{v^{1},p}} \leq \Phi(R_{r}^{\min}(\ell)) \cdot \mathcal{R}_{r}^{v} \leq \Phi(R^{\min}(\ell)) \cdot \mathcal{R}_{r}^{v}.
\end{align*}
As this holds for any $p$, we have $\mathcal{K}_{r}^{v} \leq \Phi(R_{r}^{\min}(\ell)) \cdot \mathcal{R}_{r}^{v}$. The reverse inequality can be proved analogously.
\end{proof}
Furthermore, we must show that for $\ell_{0} = \Theta(1/\delta)$, \cref{prop:messagedecayrate} is applicable. For this, we appeal to \cref{prop:precisessm} and the fact that $\eta^{*} \leq O(1/\Delta)$. Observe that $\sqrt{1 - \delta}^{\ell_{0} - 2} \cdot \eta^{*} \leq 1/2$ holds for $\ell_{0} = \Theta(1/\delta)$.

With these in hand, we may apply \cref{prop:messagedecayrate}, \cref{prop:precisessm} and \cref{lem:messageratiosrelate} for $\ell > \ell_{0}$ to deduce that
\begin{align*}
    \sum_{v \in L_{r}(\ell)} \mathcal{R}_{r}^{v} &\leq \frac{1}{\Phi(R^{\max}(\ell))} \cdot \sum_{v \in L_{r}(\ell)} \mathcal{K}_{r}^{v} \tag*{(\cref{lem:messageratiosrelate})} \\
    &\leq \frac{1}{\Phi(R^{\max}(\ell))} \cdot \max_{u \in L_{r}(1)} \wrapc{\sum_{v \in L_{u}(\ell-1)} \mathcal{K}_{u}^{v}} \cdot \sqrt{1 - \delta} \cdot \wrapp{1 + 2\eta^{*} \sqrt{1 - \delta}^{\ell-3}}^{\Delta+1} \tag*{(\cref{prop:messagedecayrate,prop:precisessm})} \\
    &\leq \dots \tag*{(Induction)} \\
    &\leq \frac{1}{\Phi(R^{\max}(\ell))} \cdot \max_{u \in L_{r}(\ell - \ell_{0})} \wrapc{\sum_{v \in L_{u}(\ell_{0})} \mathcal{K}_{u}^{v}} \cdot \sqrt{1 - \delta}^{\ell - \ell_{0}} \prod_{k=\ell_{0}}^{\ell-1} \wrapp{1 + 2\eta^{*}\sqrt{1-\delta}^{k-2}}^{\Delta+1} \tag*{(\cref{prop:messagedecayrate,prop:precisessm})} \\
    &\leq \frac{\Phi(R^{\min}(\ell_{0}))}{\Phi(R^{\max}(\ell))} \cdot \max_{u \in L_{r}(\ell - \ell_{0})} \wrapc{\sum_{v \in L_{u}(\ell_{0})} \mathcal{R}_{u}^{v}} \cdot \sqrt{1 - \delta}^{\ell - \ell_{0}} \prod_{k=\ell_{0}}^{\ell-1} \wrapp{1 + 2\eta^{*}\sqrt{1-\delta}^{k-2}}^{\Delta+1} \tag*{(\cref{lem:messageratiosrelate})} \\
    &\leq \frac{\Phi(R^{\min}(2))}{\Phi(R^{\max}(2))} \cdot \max_{u \in L_{r}(\ell - \ell_{0})} \wrapc{\sum_{v \in L_{u}(\ell_{0})} \mathcal{R}_{u}^{v}} \cdot \sqrt{1 - \delta}^{\ell - \ell_{0}} \exp\wrapp{O(1) \cdot \sum_{k=\ell_{0}}^{\ell-1} \sqrt{1-\delta}^{k-2}} \tag*{(Using $1 + x \leq e^{x}$, \cref{eq:etastarbound,fact:RminRmaxbounds}, and Monotonicity of $\Phi$)} \\
    &\leq \frac{R^{\max}(2)}{R^{\min}(2)} \cdot \max_{u \in L_{r}(\ell - \ell_{0})} \wrapc{\sum_{v \in L_{u}(\ell_{0})} \mathcal{R}_{u}^{v}} \cdot \sqrt{1 - \delta}^{\ell - \ell_{0}} \cdot O(1) \tag*{(Using $\frac{\Phi(R_{0})}{\Phi(R_{1})} =\sqrt{\frac{R_{1}(R_{1} + 1)}{R_{0}(R_{0} + 1)}} \leq \frac{R_{1}}{R_{0}}$ for $R_{0} \leq R_{1}$, and $\ell_{0} = \Theta(1/\delta)$)} \\
    &\leq \frac{\lambda}{\frac{\lambda}{(1 + \lambda)^{\Delta}}} \cdot \max_{u \in L_{r}(\ell - \ell_{0})} \wrapc{\sum_{v \in L_{u}(\ell_{0})} \mathcal{R}_{u}^{v}} \cdot \sqrt{1 - \delta}^{\ell - \ell_{0}} \cdot O(1) \tag*{(\cref{fact:RminRmaxbounds})} \\
    &\leq O(1) \cdot \sqrt{1 - \delta}^{\ell - \ell_{0}} \cdot \max_{u \in L_{r}(\ell - \ell_{0})} \wrapc{\sum_{v \in L_{u}(\ell_{0})} \mathcal{R}_{u}^{v}}. \tag*{(Using $\lambda \leq O(1/\Delta)$)} \\
\end{align*}

At this point, all that is left is to prove \cref{prop:messagedecayrate} and \cref{prop:precisessm}. We prove \cref{prop:messagedecayrate} and \cref{prop:precisessm} in the following subsections.
\subsection{The \texorpdfstring{$\varphi$-Pseudoinfluence}{PotentialPseudoinfluence} Decay: Proof of \texorpdfstring{\cref{prop:messagedecayrate}}{PropMessageDecay}}
Our goal in this subsection is to prove \cref{prop:messagedecayrate}. While initially this appears to be a more daunting task, it is made feasible by the fact that the tree recurrence $F$ for $R$ induces a corresponding tree recurrence for $K$ given by
\begin{align*}
    K_{r}^{\sigma} = (\varphi \circ F \circ \varphi^{-1})(K_{u}^{\sigma} : u \in L_{r}(1)).
\end{align*}
Using this tree recurrence for $K_{r}^{\sigma}$, we prove \cref{lem:truedecay} and \cref{lem:truetoideal}. Chained together with \cref{lem:idealdecay}, which lies at the heart of the results in \cite{LLY13}, we immediately obtain \cref{prop:messagedecayrate}.

Throughout, we will let $\mathbf{R} = (R_{u} : u \in L_{r}(1))$, $\mathbf{R}^{\max}(\ell) = (R_{u}^{\max}(\ell-1) : u \in L_{r}(1))$ and $\mathbf{R}^{\min}(\ell) = (R_{u}^{\min}(\ell-1) : u \in L_{r}(1))$ denote vectors with $|L_{r}(1)|$ many entries. We define $\mathbf{K}, \mathbf{K}^{\max}(\ell),\mathbf{K}^{\min}(\ell)$ analogously. Finally, if $\mathbf{x},\mathbf{y}$ are two vectors of the same dimension, then we write $\mathbf{x} \leq \mathbf{y}$ for entrywise inequality; if $y \in [-\infty,\infty]$, we write $\mathbf{x} \leq y$ if all entries of $x$ are upper bounded by $y$.
\begin{lemma}[True Decay]\label{lem:truedecay}
For every $\lambda$, and every tree $T$ rooted at $r$, we have the inequality
\begin{align*}
    \frac{\sum_{v \in L_{r}(\ell)} \mathcal{K}_{r}^{v}}{\max_{u \in L_{r}(1)}\wrapc{\sum_{v \in L_{u}(\ell-1)} \mathcal{K}_{u}^{v}}} \leq \sum_{u \in L_{r}(1)} \max_{\mathbf{K}^{\min}(\ell) \leq \mathbf{K} \leq \mathbf{K}^{\max}(\ell)} \abs{\partial_{K_{u}} (\varphi \circ F \circ \varphi^{-1})(\mathbf{K})}.
\end{align*}
\end{lemma}
\begin{lemma}[Ideal Decay; \cite{LLY13} Lemmas 12, 13, 14]\label{lem:idealdecay}
Assume $\lambda$ is up-to-$\Delta$ unique with gap $0 < \delta < 1$. Let $T$ be any tree of maximum degree $\leq\Delta$ rooted at $r$. Then we have the bound
\begin{align*}
    \max_{0 \leq \mathbf{K} \leq \infty} \norm{\nabla (\varphi \circ F \circ \varphi^{-1})(\mathbf{K})}_{1} \leq \sqrt{1 - \delta}
\end{align*}
\end{lemma}
\begin{lemma}[Relating True Decay to Ideal Decay]\label{lem:truetoideal}
Assume $\abs{R_{u}^{\max}(\ell-1) - R_{u}^{\min}(\ell-1)} \leq \eta$ for all $u \in L_{r}(1)$, where $\eta \leq \frac{1}{2}$. Then for every $\lambda$, and every tree $T$ with maximum degree $\leq\Delta$ rooted at $r$, we have the inequality
\begin{align*}
    \sum_{u \in L_{r}(1)} \max_{\mathbf{K}^{\min}(\ell) \leq \mathbf{K} \leq \mathbf{K}^{\max}(\ell)} \abs{\partial_{K_{u}} (\varphi \circ F \circ \varphi^{-1})(\mathbf{K})} \leq \wrapp{1 + 2\eta}^{\Delta+1} \norm{\nabla (\varphi \circ F \circ \varphi^{-1})(\mathbf{K}^{\max}(\ell))}_{1}
\end{align*}
\end{lemma}
\begin{proof}[Proof of \cref{lem:truedecay}]
To prove the claim, it suffices to show that if $v \in L_{u}(\ell-1)$ for $u \in L_{r}(1)$, then
\begin{align}\label{eq:termwiseineq}
    \mathcal{K}_{r}^{v} \leq \max_{\mathbf{K}^{\min}(\ell) \leq \mathbf{K} \leq \mathbf{K}^{\max}(\ell)} \abs{\partial_{K_{u}} (\varphi \circ F \circ \varphi^{-1})(\mathbf{K})} \cdot \mathcal{K}_{u}^{v}
\end{align}
since it then follows that
\begin{align*}
    \sum_{v \in L_{r}(\ell)} \mathcal{K}_{r}^{v} &= \sum_{u \in L_{r}(1)} \sum_{v \in L_{u}(\ell-1)} \mathcal{K}_{r}^{v} \\
    &\leq \sum_{u \in L_{r}(1)} \sum_{v \in L_{u}(\ell-1)} \max_{\mathbf{K}^{\min}(\ell) \leq \mathbf{K} \leq \mathbf{K}^{\max}(\ell)} \abs{\partial_{K_{u}} (\varphi \circ F \circ \varphi^{-1})(\mathbf{K})} \cdot \mathcal{K}_{u}^{v} \\
    &\leq \sum_{u \in L_{r}(1)} \wrapb{\max_{\mathbf{K}^{\min}(\ell) \leq \mathbf{K} \leq \mathbf{K}^{\max}(\ell)} \abs{\partial_{K_{u}} (\varphi \circ F \circ \varphi^{-1})(\mathbf{K})}} \cdot \wrapb{\sum_{v \in L_{u}(\ell-1)} \mathcal{K}_{u}^{v}} \\
    &\leq \wrapb{\sum_{u \in L_{r}(1)} \max_{\mathbf{K}^{\min}(\ell) \leq \mathbf{K} \leq \mathbf{K}^{\max}(\ell)} \abs{\partial_{K_{u}} (\varphi \circ F \circ \varphi^{-1})(\mathbf{K})}} \cdot \max_{u \in L_{r}(1)}\wrapc{\sum_{v \in L_{u}(\ell-1)} \mathcal{K}_{u}^{v}}
\end{align*}
as desired. Now, it remains to prove \cref{eq:termwiseineq}.

Fix an arbitrary partial assignment $p : L_{r}(\ell) \smallsetminus \{v\} \rightarrow [0,1]$. By the Mean Value Theorem, there exists $\mathbf{K}^{\min}(\ell) \leq \tilde{\mathbf{K}} \leq \mathbf{K}^{\max}(\ell)$ such that
\begin{align*}
    \mathcal{K}_{r}^{v,p} &= \abs{\partial_{K_{u}} (\varphi \circ F \circ \varphi^{-1})(\tilde{\mathbf{K}})} \cdot \mathcal{K}_{u}^{v,p} \\
    &\leq \abs{\partial_{K_{u}} (\varphi \circ F \circ \varphi^{-1})(\tilde{\mathbf{K}})} \cdot \mathcal{K}_{u}^{v} \\
    &\leq \max_{\mathbf{K}^{\min}(\ell) \leq \mathbf{K} \leq \mathbf{K}^{\max}(\ell)} \abs{\partial_{K_{u}} (\varphi \circ F \circ \varphi^{-1})(\mathbf{K})} \cdot \mathcal{K}_{u}^{v}
\end{align*}
Since this holds for all $p$, we obtain the desired bound.
\end{proof}
\begin{proof}[Proof of \cref{lem:truetoideal}]
Fix $u \in L_{r}(1)$. We have by the Chain Rule that
\begin{align*}
    \partial_{K_{u}}(\varphi \circ F \circ \varphi^{-1})(\mathbf{K}) &= (\varphi' \circ F \circ \varphi^{-1})(\mathbf{K}) \cdot (\partial_{R_{u}}F \circ \varphi^{-1})(\mathbf{K}) \cdot (\varphi^{-1})'(K_{u}) \\
    &= (\Phi \circ F)(\mathbf{R}) \cdot (\partial_{R_{u}}F)(\mathbf{R}) \cdot \frac{1}{\Phi(R_{u})}
\end{align*}
where we recall $\varphi' = \Phi$. Note that $(\varphi^{-1})' = \frac{1}{\Phi \circ \varphi^{-1}}$ follows by the Inverse Function Theorem.

We now bound each term separately under the restriction $\mathbf{R}^{\min}(\ell) \leq \mathbf{R} \leq \mathbf{R}^{\max}(\ell)$. We claim the following.
\begin{enumerate}
    \item $\abs{(\Phi \circ F)(\mathbf{R})} \leq \abs{(\Phi \circ F)(\mathbf{R}^{\max}(\ell))}$: To see this, observe that $F$ is monotone decreasing in each coordinate. Furthermore, $\Phi$ is monotone decreasing. Hence, $\Phi \circ F$ is monotone increasing in each coordinate.
    \item $\abs{\partial_{R_{u}}F(\mathbf{R})} \leq \wrapp{1 + 2\eta}^{\Delta+1} \abs{\partial_{R_{u}}F(\mathbf{R}^{\max}(\ell))}$: To see this, observe that
    \begin{align*}
        \partial_{R_{u}}F(\mathbf{R}) &= -\lambda \prod_{w \in L_{r}(1) : w \neq u} \frac{1}{R_{w} + 1} \cdot \frac{1}{(R_{u} + 1)^{2}}
    \end{align*}
    is negative and monotone increasing. Hence, $\abs{\partial_{R_{u}}F(\mathbf{R})}$ is positive and monotone decreasing. With this observation, define $\bm{\eta} = \mathbf{R}^{\max}(\ell) - \mathbf{R}^{\min}(\ell) = (\eta_{u} : u \in L_{r}(1))$ for convenience. Note that $\bm{\eta} \leq \eta$. Then we have
    \begin{align*}
        \abs{\partial_{R_{u}}F(\mathbf{R})} &\leq \abs{\partial_{R_{u}}F(\mathbf{R}^{\min}(\ell))} = \abs{\partial_{R_{u}}F(\mathbf{R}^{\max}(\ell) - \bm{\eta})} \\&= \lambda \prod_{w \in L_{r}(1) : w \neq u} \frac{1}{(R_{w}^{\max} + 1) - \eta_{w}}  \cdot \frac{1}{((R_{u}^{\max} + 1) - \eta_{u})((R_{u}^{\max} + 1) - \eta_{u})}
    \end{align*}
    Our goal is to control this latter inequality by $\wrapp{1 + 2\eta}^{\Delta+1}\abs{\partial_{R_{u}}F(\mathbf{R}^{\max}(\ell))}$. To do this, we use the following claim, which we prove at the end of this subsection.
    \begin{claim}\label{claim:etalipschitz}
    Assume $\eta \leq \frac{1}{2}$. Then for every $x \geq 0$, we have
    \begin{align*}
        \frac{1}{(x + 1) - \eta} &\leq \wrapp{1 + 2\eta} \cdot \frac{1}{x + 1}
    \end{align*}
    \end{claim}
    \begin{proof}
    Rearranging, the claim is equivalent to
    \begin{align*}
        &x + 1 \leq (1 + 2\eta)((x + 1) - \eta) = ((x + 1) - \eta) + 2\eta((x + 1) - \eta) \\
        &\iff \eta \leq 2\eta((x + 1) - \eta) \\
        &\iff \eta \leq \frac{1}{2} + x
    \end{align*}
    \end{proof}
    With this claim in hand, we see that
    \begin{align*}
        &\lambda \prod_{w \in L_{r}(1) : w \neq u} \frac{1}{(R_{w}^{\max} + 1) - \eta_{w}}  \cdot \frac{1}{((R_{u}^{\max} + 1) - \eta_{u})((R_{u}^{\max} + 1) - \eta_{u})} \\
        &\leq \underset{\text{now recall } \bm{\eta} \leq \eta}{\underbrace{\wrapp{1 + 2 \eta_{u}}^{2} \prod_{w \in L_{r}(1) : w \neq u} \wrapp{1 + 2 \eta_{w}}}} \cdot \underset{= \abs{\partial_{R_{u}}F(\mathbf{R}^{\max}(\ell))}}{\underbrace{\lambda \prod_{w \in L_{r}(1) : w \neq u} \frac{1}{R_{w}^{\max} + 1} \cdot \frac{1}{(R_{u}^{\max} + 1)^{2}}}} \\
        &\leq \wrapp{1 + 2\eta}^{\Delta+1} \abs{\partial_{R_{u}}F(\mathbf{R}^{\max}(\ell))}
    \end{align*}
    \item $\abs{\frac{1}{\Phi(R_{u})}} \leq \abs{\frac{1}{\Phi(R_{u}^{\max})}}$: This just follows by the fact that $\Phi$ is positive and monotone decreasing, so that $\frac{1}{\Phi}$ is positive and monotone increasing.
\end{enumerate}
From this, we obtain
\begin{align*}
    \max_{\mathbf{K}^{\min}(\ell) \leq \mathbf{K} \leq \mathbf{K}^{\max}(\ell)} \abs{\partial_{K_{u}} (\varphi \circ F \circ \varphi^{-1})(\mathbf{K})} \leq \wrapp{1 + 2\eta}^{\Delta+1}\cdot \abs{\partial_{K_{u}}(\varphi \circ F \circ \varphi^{-1})(\mathbf{K^{\max}})}
\end{align*}
Hence
\begin{align*}
    &\sum_{u \in L_{r}(1)} \max_{\mathbf{K}^{\min}(\ell) \leq \mathbf{K} \leq \mathbf{K}^{\max}(\ell)} \abs{\partial_{K_{u}} (\varphi \circ F \circ \varphi^{-1})(\mathbf{K})} \\
    &\leq \wrapp{1 + 2\eta}^{\Delta+1} \sum_{u \in L_{r}(1)} \abs{\partial_{K_{u}}(\varphi \circ F \circ \varphi^{-1})(\mathbf{K}^{\max}(\ell))} \\
    &= \wrapp{1 + 2\eta}^{\Delta+1} \norm{\nabla(\varphi \circ F \circ \varphi^{-1})(\mathbf{K}^{\max}(\ell))}_{1}
\end{align*}
as desired.
\end{proof}
\begin{remark}
We note that the proofs of \cref{lem:truedecay} and \cref{lem:truetoideal} did not truly rely on the fact that $\Phi(R)$ had the form $\frac{1}{\sqrt{R(R + 1)}}$. The arguments go through for any continuously differentiable, monotone increasing, concave potential function. Where we needed the definition of $\Phi$ itself is in the bound on $\norm{\nabla (\varphi \circ F \circ \varphi^{-1})(\mathbf{K})}_{1}$ given in \cref{lem:idealdecay}, which was proved in \cite{LLY13}.
\end{remark}

\section{Conclusion and Open Problems}
In this work we have shown that for the hardcore distribution on independent sets of a graph of maximum degree $\leq\Delta$ with parameter $\lambda = (1 - \delta)\lambda_{c}(\Delta)$, there is a constant $C(\delta)$ such that the Glauber dynamics mixes in $O(n^{C(\delta)})$ steps. While this running time does not have an exponential dependence on $\log \Delta$ as in the correlation decay algorithm of \cite{Wei06}, its dependence on $\delta$ is significantly worse. Specifically, we have that $C(\delta) \leq \exp\wrapp{O(1/\delta)}$, while the correlation decay algorithm of \cite{Wei06} exhibits a dependence of $C(\delta) \leq O(1/\delta)$.

In a follow-up work by \cite{CLV20}, they showed how one can bound the total pairwise influence of the root of a tree on all other vertices. This is in contrast to our analysis, which focuses on the total pairwise influence of all other vertices on the root. They achieve an upper bound of $O(1/\delta)$, giving $O(1/\delta)$-spectral independence and $n^{O(1/\delta)}$ mixing. They also generalize to all antiferromagnetic two-state spin systems. We leave it as an open problem to bound the total pairwise influence on the root by $O(1/\delta)$, and to generalize this to other two-state spin systems.

We show in \cref{app:infonregtree} that in general, one cannot bound $\lambda_{\max}(\Psi_{\mu})$ asymptotically better than $O(1/\delta)$, even for the special case of trees. We do this by showing for the infinite $\Delta$-regular tree that the total pairwise influence on a vertex is $\Theta(1/\delta)$. This shows that in general the best bound on the mixing time of the Glauber dynamics one can hope to achieve by bounding the spectral independence and applying the local-to-global theorem of \cite{AL20} is $n^{O(1/\delta)}$. However, prior results \cite{Vig01, EHSVY16} for this problem appear to suggest that $O\wrapp{C(\delta) n\log n}$ should be possible, which illustrates a key limitation of the current local-to-global results.

\pagebreak
\printbibliography
\pagebreak

\begin{appendices}
\crefalias{section}{appsec}

\section{Precise Strong Spatial Mixing: Proof of \texorpdfstring{\cref{prop:precisessm}}{PropPreciseSSM}}\label{sec:precisessm}
In this subsection, our goal is to prove \cref{prop:precisessm}. We use the following strong spatial mixing result proved in \cite{LLY13}.
\begin{theorem}[Theorem 9 from \cite{LLY13}]\label{thm:ssmK}
Assume that $\lambda$ is up-to-$\Delta$ unique with rate $\delta$, that is, $\delta = 1 - \max_{1 \leq d < \Delta} \abs{f_{d}'(\hat{R}_{d})}$ satisfies $0 < \delta < 1$. For every $T$ rooted at $r$ and every level $\ell$, we have
\begin{align*}
    \abs{K_{r}^{\min}(\ell) - K_{r}^{\max}(\ell)} \leq \sqrt{1 - \delta} \cdot \max_{u \in L_{r}(1)} \abs{K_{u}^{\min}(\ell-1) - K_{u}^{\max}(\ell-1)}
\end{align*}
\end{theorem}
\begin{proof}[Proof of \cref{prop:precisessm}]
First, observe that
\begin{align}\label{eq:RssmKssm}
    \Phi(R^{\max}(\ell)) \cdot \abs{R_{r}^{\min}(\ell) - R_{r}^{\max}(\ell)} \leq \abs{K_{r}^{\min}(\ell) - K_{r}^{\max}(\ell)} \leq \Phi(R^{\min}(\ell)) \cdot \abs{R_{r}^{\min}(\ell) - R_{r}^{\max}(\ell)}
\end{align}
This holds via a nearly identical argument to the proof of \cref{lem:messageratiosrelate}. With these inequalities in hand, we have
\begin{align*}
    \abs{R_{r}^{\min}(\ell) - R_{r}^{\max}(\ell)} &\leq \frac{1}{\Phi(R^{\max}(\ell)} \cdot \abs{K_{r}^{\min}(\ell) - K_{r}^{\max}(\ell)} \tag*{(\cref{eq:RssmKssm})} \\
    &\leq \frac{1}{\Phi(R^{\max}(\ell))} \cdot \sqrt{1 - \delta}^{\ell-2} \cdot \max_{u \in L_{r}(\ell - 2)}\wrapc{\abs{K_{u}^{\min}(2) - K_{u}^{\max}(2)}} \tag*{(\cref{thm:ssmK})} \\
    &\leq \frac{\Phi(R^{\min}(2))}{\Phi(R^{\max}(\ell))} \cdot \sqrt{1 - \delta}^{\ell-2} \cdot \max_{u \in L_{r}(\ell - 2)}\wrapc{\abs{R_{u}^{\min}(2) - R_{u}^{\max}(2)}} \tag*{(\cref{eq:RssmKssm})} \\
    &\leq \frac{\Phi(R^{\min}(2))}{\Phi(R^{\max}(2))} \cdot \sqrt{1 - \delta}^{\ell-2} \cdot \max_{u \in L_{r}(\ell - 2)}\wrapc{\abs{R_{u}^{\min}(2) - R_{u}^{\max}(2)}} \tag*{(\cref{fact:RminRmaxbounds})} \\
    &\leq \frac{R^{\max}(2)}{R^{\min}(2)} \cdot \sqrt{1 - \delta}^{\ell-2} \cdot \max_{u \in L_{r}(\ell - 2)}\wrapc{\abs{R_{u}^{\min}(2) - R_{u}^{\max}(2)}} \tag*{($\frac{\Phi(R_{0})}{\Phi(R_{1})} =\sqrt{\frac{R_{1}(R_{1} + 1)}{R_{0}(R_{0} + 1)}} \leq \frac{R_{1}}{R_{0}}$ for $R_{0} \leq R_{1}$)} \\
    &\leq \sqrt{1 - \delta}^{\ell - 2} \cdot \eta^{*} \tag*{(\cref{def:etas})}
\end{align*}
\end{proof}

\section{The Pairwise Influence Matrix on Infinite Regular Trees}\label{app:infonregtree}
In this section, we consider general two-state spin systems, and analyze $\Psi_{\mu}$ for the infinite $\Delta$-regular tree in the uniqueness regime. A two-state spin system is specified by a nonnegative symmetric matrix of edge activities $A = \begin{bmatrix} A_{0,0} & A_{0,1} \\ A_{1,0} & A_{1,1} \end{bmatrix} = \begin{bmatrix} \beta & 1 \\ 1 & \gamma \end{bmatrix}$, and a vector of vertex activities $b = \begin{bmatrix} b_{0} \\ b_{1} \end{bmatrix} = \begin{bmatrix} \lambda \\ 1 \end{bmatrix}$. For an input graph $G=(V,E)$, the parameters $\beta,\gamma,\lambda$ induce a distribution $\mu$ over assignments (or configurations) $\sigma:V \rightarrow \{0,1\}$ given by
\begin{align*}
    \mu(\sigma) \propto \beta^{m_{0}(\sigma)} \gamma^{m_{1}(\sigma)} \lambda^{|\sigma^{-1}(0)|}
\end{align*}
where $m_{0}(\sigma)$ is the number of edges $\{u,v\} \in E$ such that $\sigma(u) = \sigma(v) = 0$, and $m_{1}(\sigma)$ is the number of edges $\{u,v\} \in E$ such that $\sigma(u) = \sigma(v) = 1$. Note it is without loss of generality that we take $A_{0,1} = A_{1,0} = b_{1} = 1$, as we may always scale the partition function by a universal constant. We also assume without loss of generality that $\beta \leq \gamma$.

The partition function corresponding to a two-state spin system on $G=(V,E)$ with parameters $\beta,\gamma,\lambda$ is given by
\begin{align*}
		Z_{G}(\beta,\gamma,\lambda) = \sum_{\sigma:V \rightarrow \{0,1\}} \beta^{m_{0}(\sigma)}\gamma^{m_{1}(\sigma)}\lambda^{|\sigma^{-1}(0)|}
\end{align*}
If $\beta\gamma < 1$, then the system is considered antiferromagnetic. Otherwise, the system is ferromagnetic. In the case $\beta = 0$ and $\gamma = 1$, we recover the hardcore model. Another interesting case is when $\beta = \gamma$, which corresponds to the Ising model.

Here, the generalization of the tree recurrence for the hardcore model \cref{eq:treerecurrence} is given by
\begin{align}\label{eq:generaltreerecurrence}
    f_{\Delta-1}(R) = \lambda \wrapp{\frac{\beta R + 1}{R + \gamma}}^{\Delta-1}
\end{align}
For convenience, we will write $f$ instead of $f_{\Delta-1}$, and $\hat{R}$ instead of $\hat{R}_{\Delta-1}$.

In this section, we show that for parameters $(\beta,\gamma,\lambda)$ in the uniqueness regime for the infinite $\Delta$-regular tree, the corresponding pairwise influence matrix of the system satisfies $\lambda_{\max}(\Psi_{\mu}) = \Theta(1/\delta)$, where $\delta = 1 - \abs{f_{\Delta-1}'(\hat{R}_{\Delta-1})}$. The main result of this section is the following.
\begin{theorem}\label{thm:infinitetreebound}
Assume the parameters $\beta,\gamma,\lambda$ are within the uniqueness regime of $\mathbb{T}_{\Delta}$, and let $\delta = 1- \abs{f'(\hat{R})}$. Then we have the bounds
\begin{align*}
    \frac{\Delta}{\Delta-1} \cdot \wrapp{\frac{1}{\delta(2-\delta)} - 1} \leq \lambda_{\max}(\Psi_{\mu}) \leq \frac{\Delta}{\Delta-1} \cdot \wrapp{\frac{1}{\delta} - 1}
\end{align*}
\end{theorem}
We prove this using the following two lemmas. The first concerns the structure of $\Psi_{\mu}$, which holds for arbitrary trees. The second is a calculation of the entries of $\Psi_{\mu}$ for $\mathbb{T}_{\Delta}$.
\begin{lemma}\label{lem:treeinfprodproperty}
Let $\mu$ denote the distribution corresponding to a two-state spin system on an arbitrary $T$ with arbitrary parameters $\beta,\gamma,\lambda$. Let $u,v,w$ be distinct vertices in $T$ such that $w$ is on the unique path from $u$ to $v$. Then $\Psi_{\mu}(u,v) = \Psi_{\mu}(u,w) \cdot \Psi_{\mu}(w,v)$.
\end{lemma}
\begin{lemma}\label{lem:infinitetreeedgecorrelation}
Assume the parameters $\beta,\gamma,\lambda,\Delta$ are within the uniqueness regime of $\mathbb{T}_{\Delta}$. If $u \sim v$ in $\mathbb{T}_{\Delta}$, then
\begin{align*}
    \Psi_{\mu}(u,v) = \frac{f'(\hat{R})}{\Delta-1}
\end{align*}
\end{lemma}
Assuming the truth of these two lemmas, we may immediately prove \cref{thm:infinitetreebound}.
\begin{proof}[Proof of \cref{thm:infinitetreebound}]
Fix an arbitrary root node $r$ in $\mathbb{T}_{\Delta}$. Combining \cref{lem:treeinfprodproperty} and\cref{lem:infinitetreeedgecorrelation}, we have that $\Psi_{\mu}(v,r) = \wrapp{\frac{f'(\hat{R})}{\Delta-1}}^{d(v,r)}$, where $d(v,r)$ denotes the shortest path distance between $v$ and $r$. Note that by \cref{lem:infinitetreeedgecorrelation} and the fact that $f'(\hat{R}) \leq 0$ when the spin system is antiferromagnetic, $\Psi_{\mu}(v,r)$ is negative when $d(v,r)$ is odd, and positive otherwise.

With this observation in hand, we see that the upper bound follows by the following calculation:
\begin{align*}
    \lambda_{\max}(\Psi_{\mu}) &\leq \sum_{v \in \mathbb{T}_{\Delta} : v \neq r} \abs{\Psi_{\mu}(v,r)} \\
    &= \sum_{k=1}^{\infty} \sum_{v \in \mathbb{T}_{\Delta} : d(v,r) = k} \abs{\Psi_{\mu}(v,r)} \\
    &= \frac{\Delta}{\Delta-1} \sum_{k=1}^{\infty} (\Delta-1)^{k} \wrapp{\frac{|f_{\beta,\gamma,\lambda}'(\hat{R})|}{\Delta-1}}^{k} \\
    &= \frac{\Delta}{\Delta-1} \sum_{k=1}^{\infty} (1 - \delta)^{k} \\
    &= \frac{\Delta}{\Delta-1} \wrapp{\frac{1}{\delta} - 1}
\end{align*}
For the lower bound, let $S$ denote the subset of vertices of even distance from $r$ (including $r$ itself), and let $\Psi_{\mu}^{S,S}$ denote the $S \times S$ principal submatrix of $\Psi_{\mu}$. Since all vertices of $S$ have even distance from each other, all entries of $\Psi_{\mu}^{S,S}$ are positive. Furthermore, by symmetric, all rows of $\Psi_{\mu}$ have the same sum. Hence, we have
\begin{align*}
    \lambda_{\max}(\Psi_{\mu}) &\geq \lambda_{\max}(\Psi_{\mu}^{S,S}) \geq \sum_{v \in S : v \neq r} \Psi_{\mu}^{S,S}(v,r) = \frac{\Delta}{\Delta-1} \sum_{k=1}^{\infty} (1 - \delta)^{2k} = \frac{\Delta}{\Delta-1} \wrapp{\frac{1}{\delta(2-\delta)} - 1}
\end{align*}
gives the desired lower bound.
\end{proof}
It remains to prove the two lemmas.
\begin{proof}[Proof of \cref{lem:treeinfprodproperty}]
This lemma follows by conditional independence of $v$ and $u$ when the spin of $w$ is fixed, and the Law of Conditional Probability. Specifically, we have
\begin{align*}
    \Pr[v \mid u] &= \Pr[v, w \mid u] + \Pr[v,\overline{w} \mid u] = \Pr[v \mid u,w] \cdot \Pr[w \mid u] + \Pr[v \mid u,\overline{w}] \cdot \Pr[\overline{w} \mid u] \\
    &= \Pr[v \mid w] \cdot \Pr[w \mid u] + \Pr[v \mid \overline{w}] \cdot \Pr[\overline{w} \mid u]
\end{align*}
Similarly,
\begin{align*}
    \Pr[v \mid \overline{u}] = \Pr[v \mid w] \cdot \Pr[w \mid \overline{u}] + \Pr[v \mid \overline{w}] \cdot \Pr[\overline{w} \mid \overline{u}]
\end{align*}
Here, the second step follows by the Law of Condition Probability, and the final step follows by conditional independence of $u,v$ given $w$.
\begin{align*}
    \Psi_{\mu}(u,v) &= \Pr[v \mid u] - \Pr[v \mid \overline{u}] \\
    &= \Pr[v \mid w] \cdot (\underset{=\Psi_{\mu}(u,w)}{\underbrace{\Pr[w \mid u] - \Pr[w \mid \overline{u}]}}) + \Pr[v \mid \overline{w}] \cdot (\underset{=-\Psi_{\mu}(u,w)}{\underbrace{\Pr[\overline{w} \mid u] - \Pr[\overline{w} \mid \overline{u}}}) \\
    &= \Psi_{\mu}(u,w) \cdot \Psi_{\mu}(w,v)
\end{align*}
\end{proof}
\begin{proof}[Proof of \cref{lem:infinitetreeedgecorrelation}]
We calculate that $R[v \mid u] = \beta \lambda \prod_{w \sim v : w \neq u} \frac{\beta \hat{R} + 1}{\hat{R} + \gamma} = \beta f_{\beta,\lambda,\gamma}(\hat{R}) = \beta \hat{R}$ and similarly, $R[v\mid \overline{u}] = \frac{1}{\gamma}\lambda \prod_{w \sim v : w \neq u} \frac{\beta \hat{R} + 1}{\hat{R} + \gamma} = \frac{1}{\gamma}f_{\beta,\gamma,\lambda}(\hat{R}) = \frac{1}{\gamma}\hat{R}$. Hence, we have that
\begin{align*}
    \Psi_{\mu}(u,v) &= \Pr[v \mid u] - \Pr[v \mid \overline{u}] = \frac{R[v \mid u]}{R[v \mid u] + 1} - \frac{R[v \mid \overline{u}]}{R[v \mid \overline{u}] + 1} \\
    &= \frac{\beta \hat{R}}{\beta \hat{R} + 1} - \frac{\hat{R}}{\hat{R} + \gamma} = \frac{(\beta \gamma - 1)\cdot \hat{R}}{(\beta \hat{R} + 1)(\hat{R} + \gamma)}
\end{align*}
Now, we compute that
\begin{align*}
    f'(R) = \frac{(\Delta-1)(\beta \gamma - 1)}{(\beta R + 1)(R + \gamma)} \cdot \lambda \wrapp{\frac{\beta R + 1}{R + \gamma}}^{\Delta-1} = (\Delta-1) \frac{\beta\gamma - 1}{(\beta R + 1)(R + \gamma)} f(R)
\end{align*}
At $\hat{R}$, we have $f(\hat{R}) = \hat{R}$ so that
\begin{align*}
    f'(\hat{R}) = (\Delta-1) \frac{(\beta\gamma - 1) \cdot \hat{R}}{(\beta \hat{R} + 1)(\hat{R} + \gamma)} = (\Delta-1) \cdot \Psi_{\mu}(u,v)
\end{align*}
The claim follows.
\end{proof}

\section{Gapped \texorpdfstring{Up-to-$\Delta$}{Up-to-Delta} Uniqueness: Proof of \texorpdfstring{\cref{lem:gappeduptodeltaunique}}{LemmaGappedUptoDelta}}\label{sec:gappeduptodeltaunique}
\begin{lemma}[Gapped Up-to-$\Delta$ Uniqueness in the Hardcore Model]\label{lem:gappeduptodeltaunique}
$\lambda$ is up-to-$\Delta$ unique with gap $0 < \delta < 1$ if and only if $\lambda < (1 - \Theta(\delta))\lambda_{c}(\Delta)$.
\end{lemma}
\begin{proof}
Let $1 < d < \Delta$. First, we calculate that
\begin{align*}
    f_{d}'(R) = -d\lambda \cdot \wrapp{\frac{1}{R+1}}^{d+1} = -d \cdot \frac{f_{d}(R)}{R+1}
\end{align*}
In particular, at the unique fixed point of $f_{d}$, we have
\begin{align*}
    \abs{f_{d}'(\hat{R}_{d})} = d \cdot \wrapp{1 - \frac{1}{\hat{R}_{d} + 1}}
\end{align*}
Up-to-$\Delta$ uniqueness holds only if $\abs{f_{d}'(\hat{R}_{d})} \leq 1 - \delta$. In terms of $\hat{R}_{d}$, this holds if and only if $\hat{R}_{d} \leq \frac{1 - \delta}{d - 1 + \delta}$. Observe that since $f_{d}$ is monotone decreasing and $\hat{R}_{d}$ is the unique fixed point of $f_{d}$, we have $f_{d}(R) < R$ for all $R > \hat{R}_{d}$ and $f_{d}(R) > R$ for all $R < \hat{R}_{d}$. Hence, $\frac{1-\delta}{d-1+\delta} \geq \hat{R}_{d}$ holds if and only if
\begin{align*}
    \lambda \wrapp{\frac{d-1+\delta}{d}}^{d} = f_{d}\wrapp{\frac{1-\delta}{d-1+\delta}} \leq \frac{1-\delta}{d-1+\delta} \iff \lambda \leq \wrapp{\frac{d}{d-1+\delta}}^{d} \cdot \frac{1-\delta}{d-1+\delta} \overset{\defin}{=} \lambda_{c}(\delta,d+1)
\end{align*}
Now, let us compare this with $\lambda_{c}(d + 1) = \wrapp{\frac{d}{d-1}}^{d} \cdot \frac{1}{d-1}$. Define $c(\delta)$ such that $(1 - c(\delta)) \lambda_{c}(d+1) = \lambda_{c}(\delta,d+1)$. Since $\lambda_{c}(d+1)$ is monotone decreasing in $d$ and $d < \Delta$, we have $\lambda_{c}(d+1) \geq \lambda_{c}(\Delta)$. Thus, we have shown that $\lambda$ is up-to-$\Delta$ unique with gap $0 < \delta < 1$ if and only if $\lambda < (1 - c(\delta)) \lambda_{c}(\Delta)$. All that remains is to show $c(\delta) = \Theta(\delta)$.

For this, we first calculate that
\begin{align*}
    1 - c(\delta) = (1 - \delta)\wrapp{\frac{d-1}{d-1+\delta}}^{d+1}
\end{align*}
Clearly, $\wrapp{\frac{d-1}{d-1+\delta}}^{d+1} \leq 1$, which implies $c(\delta) \geq \delta$. On the other hand, by Bernoulli's Inequality, we also have that
\begin{align*}
    \wrapp{\frac{d-1}{d-1+\delta}}^{d+1} = \wrapp{1 - \frac{\delta}{d-1+\delta}}^{d+1} \geq 1 - \frac{d+1}{d-1+\delta} \cdot \delta
\end{align*}
which implies $c(\delta) \leq O(\delta)$ as well. This concludes the proof.
\end{proof}

\end{appendices}
\end{document}